\documentclass[5p,twocolumn,authoryear]{elsarticle}

\journal{Automatica}

\usepackage{xcolor}
\colorlet{RefColor}{green!50!black}
\colorlet{LinkColor}{red!50!black}

\usepackage{hyperref}
\hypersetup{
  colorlinks = true,
  linkbordercolor = {white},
  linkcolor = LinkColor,
  anchorcolor=LinkColor,
  citecolor= RefColor,
  filecolor=cyan,
  menucolor=red,
  runcolor=cyan,
  urlcolor = LinkColor,
}

\usepackage{amsmath,amssymb}
\usepackage{amsthm}
\usepackage{color}
\usepackage{xcolor}
\usepackage{enumerate}
\usepackage{graphicx}
\usepackage{siunitx}
\usepackage{tikz,pgfplots}
\usepackage[ruled,vlined]{algorithm2e}
\pgfplotsset{compat=1.16}
\usetikzlibrary{external,arrows.meta,calc,shapes,pgfplots.groupplots,
decorations.markings,graphs,spy,intersections,backgrounds, positioning,patterns}
\usepgfplotslibrary{groupplots}
\pgfplotsset{compat=1.16}
\usepackage[utf8]{inputenc}
\usepackage{mathtools}
\usepackage{url}
\usepackage{cleveref}
\usepackage{cuted}
\usepackage{subfigure}
\usepackage{mathrsfs} 
\usepackage{multirow}

\newtheorem{lem}{Lemma}
\newtheorem{thm}{Theorem}
\newtheorem{remark}{Remark}
\newtheorem{proposition}{Proposition}
\newtheorem{definition}{Definition}

\makeatletter
\newcommand{\pushright}[1]{\ifmeasuring@#1\else\omit\hfill$\displaystyle#1$\fi\ignorespaces}
\newcommand{\pushleft}[1]{\ifmeasuring@#1\else\omit$\displaystyle#1$\hfill\fi\ignorespaces}
\makeatother

\newcommand{\hinfheight}{0.17\textwidth}
\newcommand{\COMPlib}{\textit{COMPl$_e$ib} }

\usepackage{Notation}

\begin{document}

\begin{frontmatter}	
	\title{Fixed-Order H-Infinity Controller Design for Port-Hamiltonian 
Systems\tnoteref{t1,t2}}
    \tnotetext[t1]{This research has been supported by the German Research 
Foundation (DFG) within the project 424221635.}
    \tnotetext[t2]{\textbf{CRediT author statement:} \textbf{Paul Schwerdtner} Conceptualization, Methodology, Software, Data Curation, Writing -- Original Draft, Visualization, \textbf{Matthias Voigt:} Conceptualization, Writing -- Review \& Editing, Supervision, Funding Acquisition}
    
	\author[TUB]{Paul Schwerdtner\corref{COR}}
	\ead{schwerdt@math.tu-berlin.de}	
	\author[FUS]{Matthias Voigt}\ead{matthias.voigt@fernuni.ch}
	
	\affiliation[TUB]{organization={TU Berlin, Institute of Mathematics},
		addressline={Stra{\ss}e des 17.~Juni 136}, 
		city={10623 Berlin},
		country={Germany}}
	\affiliation[FUS]{organization={UniDistance Suisse},
	addressline={Schinerstrasse 18}, 
	city={3900 Brig},
	country={Switzerland}}			
	\cortext[COR]{Corresponding author.}

	\begin{abstract}
          We present a new fixed-order H-infinity controller design method for potentially large-scale port-Hamiltonian (pH) plants. Our method computes controllers that are also pH (and thus passive) such that the resulting closed-loop systems is again passive, which ensures closed-loop stability simply from the structure of the plant and controller matrices. In this way, we can avoid computationally expensive eigenvalue computations that would otherwise be necessary. In combination with a sample-based objective function which allows us to avoid multiple evaluations of the H-infinity norm (which is typically the main computational burden in fixed-order H-infinity controller synthesis), this makes our method well-suited for plants with a high state-space dimension.

          In our numerical experiments, we show that applying a passivity-enforcing post-processing step after using well-established H-infinity synthesis methods often leads to a deteriorated H-infinity performance. In constrast to that, our method computes pH controllers, that are automatically passive and simultaneously aim to minimize the H-infinity norm of the closed-loop transfer function. Moreover, our experiments show that for large-scale plants, our method is significantly faster than the well-established fixed-order H-infinity controller synthesis methods.
	\end{abstract}
	
\end{frontmatter}



\section{Introduction}

The port-Hamiltonian (pH) modeling paradigm is used for energy-based modeling of 
complex processes accross several physical domains including electrical 
systems~\citep{MehMOS18}, flow-problems~\citep{HauMMMMRS19}, and mechanical 
multi-body systems~\citep[Example 12]{BeaMXZ18}. The two main benefits of pH 
models are the intuitive energy-based interconnection of systems from different 
physical domains and the beneficial properties such as \emph{passivity} that 
follow directly from the model structure. In recent years, the increased demand 
for accurate models of complex multi-physical processes has led to a 
wide-spread utilization of pH models and subsequently system theoretical tools 
such as model order reduction and system identification have been adapted to pH 
systems; see~\cite{MehU22} for 
an overview of recent methods and applications. The available pH controller 
design strategies developed in~\cite{OrtVCA08, RamLMC16, ZhaBOLS17} are mostly based 
on energy-shaping and do not aim for classical $\htwo$ or $\hinf$ optimal 
controllers. In this work, we present a new algorithm for fixed-order $\hinf$ 
controller design for pH systems.

Well-established numerical methods for classical $\hinf$ controller design are
based on the repeated solution of linear matrix inequalities (LMIs) or algebraic 
Riccati equations (AREs); see, e.\,g., the monographs~\cite{Fra87, ZhoDG96, 
SkoP05}. The AREs are often solved using structured generalized eigenvalue 
problems~\citep{BenBMX02,BenBLMX11}. LMI- or ARE-based control is widely used for 
plants with small state-dimension (also called plant \emph{order}). The main 
drawback of these methods is that the obtained controllers have the same order 
as the plant (and are thus often called \emph{full-order} controllers). This is 
undesired for plants with higher order, since the implementation of 
high controllers is often impractical. For some system classes, low-rank 
solvers for AREs have enabled the computation of low-order 
controllers~\citep{BenHW22b}. 
Otherwise, a subsequent \emph{controller order reduction} is 
necessary~\citep{MusG91, AndL89}. Note that full-order ARE-based $\hinf$ control 
for pH systems is currently under investigation; see~\cite{Breiten2022} for 
a recent preprint.

In the early 2000s, \emph{fixed-order} $\hinf$ controller synthesis methods have 
been established, in which the set of controllers of a fixed and typically low 
order is parameterized and then the controller parameters are optimized using 
gradient-based numerical optimization. There exist two implementations of this 
approach, \texttt{HIFOO}\footnote{available 
at~\url{https://cs.nyu.edu/~overton/software/hifoo/}}~\citep{Burke2006Hifoo} and 
\texttt{hinfstruct}\footnote{available in MATLAB's Robust Control 
Toolbox}~\citep{Apkarian2006Hinfstruct}. Fixed-order $\hinf$ synthesis leads to 
constraint nonsmooth and nonconvex optimization problems such that 
only locally optimal controllers can be computed. In contrast to that, 
full-order controllers are typically globally optimal\footnote{For 
numerical reasons, full-order controllers are often designed as 
suboptimal controllers but with near globally optimal performance.}. 
However, both packages were applied successfully in industrial applications, 
see, e.\,g.,~\cite{WanC09,GabAN10,RobBBPA10,RavN12}. In each iteration of the 
optimization, the $\hinf$ norm of the closed-loop system must be evaluated and 
its stability must be checked. For systems with high state-space dimension, 
this 
can lead to prohibitively high computational costs.

To keep the computational costs manageable for large-scale systems,
\texttt{HIFOO} is extended in~\cite{MitO15} by using a new method for the 
$\hinf$ norm computation of large-scale systems~\citep{MitO16} and 
in~\cite{BenMO18} by using a reduced-order approximation to the given 
large-scale plant. Recently, in~\cite{WerOP22}, controllers are computed using 
several reduced-order plants at different fidelity levels. To ensure stability, 
in the large-scale setting, iterative methods can be employed as 
in~\cite{BenMO18} to compute the rightmost eigenvalues of the closed-loop 
system matrix. When iterative methods are used, technically, stability cannot 
be ensured, as the convergence to the rightmost eigenvalue of the closed-loop 
system matrix is not guaranteed. Moreover, in~\cite{BenMO18} it is noted that 
the stability constraint in the large-scale $\hinf$ control problem poses a 
particular challenge during optimization, which requires several restarts of the 
$\hinf$ optimization in~\cite[Algorithm~2]{BenMO18}. Moreover, performing 
large-scale $\hinf$ norm evaluations during controller synthesis as 
in~\cite{MitO15} is reported to potentially lead to a premature halt of the 
optimization~\citep[Section~2]{BenMO18}.

In this article, by using the pH paradigm, we can use a different strategy to
ensure stability of large-scale closed-loop systems based on the following 
well-known fact: The negative feedback interconnection of two pH systems (see 
\Cref{sec:pHsystems}) is 
again pH and thus Lyapunov stable. Therefore, we propose to use a pH controller 
to ensure Lyapunov stability of the closed-loop system directly from the model 
structure. In fact, since the pH structure can be encoded into our controller 
parameterization, we can use unconstrained optimization.
To avoid computing the $\hinf$ norm of the potentially 
large-scale closed-loop systems repeatedly, we use the sample-based optimization 
approach developed in~\cite{SchV22, SchV21} for $\hinf$-inspired model order 
reduction. In this way, we may compute stabilizing controllers 
for pH systems that only require the 
evaluation of the transfer function of the given plant.

Our article is organized as follows. In the next section, we recall the 
fixed-order $\hinf$ synthesis problem. After that, we briefly explain pH systems 
and the connections between pH and passive systems. In \Cref{sec:ourmethod}, we 
explain our approach and in \Cref{sec:numerics} we provide a comparison to the 
other fixed-order $\hinf$ synthesis methods.

\section{$\hinf$ Controller Synthesis} \label{sec:hinfsyn}
The setup of the classical $\hinf$ controller design problem consisting of a plant $\bP$ and a controller $\bK$ is illustrated in \Cref{fig:controller}. The plant is defined as
\begin{align*}
  \mathbf{P}:
  \begin{bmatrix}
    \dot x(t) \\ z(t) \\ y(t)
  \end{bmatrix}
  =
  \begin{bmatrix}
    A & B_1 & B_2 \\
    C_1 & D_{11} & D_{12} \\
    C_2 & D_{21} & D_{22}
  \end{bmatrix}
  \begin{bmatrix}
    x(t) \\ w(t) \\ u(t)
  \end{bmatrix}
\end{align*}
with $A \in \R^{n \times n}$, $B_j \in \R^{n \times m_j}$, $C_i \in \R^{p_i \times n}$, and $D_{ij} \in \R^{p_i \times m_j}$ for $i,\,j = 1,\,2$. For $t \in [0, \infty)$, $x$ denotes the \emph{state}, $u$ and $w$ denote the \emph{control} and \emph{disturbance} input, respectively, whereas $y$ and $z$ denote the \emph{measured} and \emph{performance} output, respectively.
A controller \emph{of order $k$} is given by
\begin{align*}
  \bK:
  \begin{bmatrix}
    \rdx(t) \\ \red{y}(t)
  \end{bmatrix}
  =
  \begin{bmatrix}
    \rA & \rB \\
    \rC & \rD
  \end{bmatrix}
  \begin{bmatrix}
    \rx(t) \\ \red{u}(t)
  \end{bmatrix}
\end{align*}
with $\rA \in \R^{\nK \times \nK}$, $\rB \in \R^{\nK \times p_2}$, $\rC \in \R^{m_2 \times \nK}$, and $\rD \in \R^{m_2 \times p_2}$, and is coupled with the plant $\bP$ by the coupling conditions $\red{u}(\cdot) = y(\cdot)$ and $u(\cdot) = \red{y}(\cdot)$.

\begin{figure}
  \centering
  \begin{tikzpicture}[>=Latex]
    \node[draw=none,fill=none] (begin) {};
    \node[draw,rectangle,minimum height = 1.2cm,minimum width = 
      2cm,thick,xshift = 3cm] (plant) at (begin) {$\bP$};
    \node[draw,rectangle,minimum height = 0.8cm,minimum width = 
      2cm,thick,yshift=-1.5cm] (controller) at (plant) {$\bK$};
    \node[draw=none,fill=none,xshift = 3cm] (end) at (plant) {};
    \draw[->,thick] ($(begin.east)+(0,0.25cm)$) -- ($(plant.west)+(0,0.25cm)$)
      node[xshift=-1.5cm,above] {$w$};
    \draw[->,thick] ($(plant.east)+(0,0.25cm)$) -- ($(end.west)+(0,0.25cm)$)
      node[xshift=-0.5cm,above] {$z$};
    \draw[->,thick] ($(plant.east)+(0,-0.25cm)$) -- 
      ($(plant.east)+(1cm,-0.25cm)$) -- ($(controller.east)+(1cm,0cm)$) -- 
      (controller.east) node[xshift=1cm,yshift=0.625cm,right] {$u_{\bK} = y$};
    \draw[->,thick] (controller.west) -- 
      ($(controller.west)+(-1cm,0cm)$) -- ($(plant.west)+(-1cm,-0.25cm)$) -- 
      ($(plant.west)+(0,-0.25cm)$) node[xshift=-1cm,yshift=-0.625cm,left] {$u 
= 
      y_{\bK}$};
  \end{tikzpicture}
  \caption{Closed-loop system obtained by coupling the plant $\bP$ with the controller $\bK$}\label{fig:controller}
\end{figure}
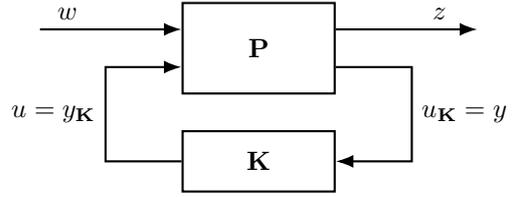

The goal in $\hinf$ synthesis is to tune the controller, i.\,e., to determine 
matrices $\rA, \rB, \rC$, and $\rD$ such that the performance output is 
minimized in an appropriate norm for all admissible disturbance inputs. This can 
be quantified using the \emph{transfer functions} of plant and controller, that 
are defind as follows
\begin{align}
 P(s) & := \begin{bmatrix}
   P_{11}(s) & P_{12}(s) \\
   P_{21}(s) & P_{22}(s) 
\end{bmatrix} \nonumber \\
      &:= \begin{bmatrix}
      C_1 \\
    C_2
  \end{bmatrix}
  {(sI_n - A)}^{-1} 
\begin{bmatrix} B_1 & B_2 \end{bmatrix}  + 
\begin{bmatrix} D_{11} & D_{12} \\ D_{21} & D_{22} \end{bmatrix}, \nonumber \\
K(s) &:= \rC{(sI_{\nK} - \rA)}^{-1} \rB + \rD. \label{eq:controller}
\end{align}

If the interconnection between $\bP$ and $\bK$ is \emph{well-posed}, i.\,e., $ 
\begin{bsmallmatrix}I_{m_2} & -D_{\bK_*} \\ -D_{22} & I_{p_2} 
\end{bsmallmatrix}$ is invertible, then the closed-loop transfer function 
exists and is given by the \emph{lower linear fractional transformation} of $P$ 
and $K$ defined as
\begin{multline*}
 (P \star K)(s) := P_{11}(s) \\ + 
P_{12}(s)K(s){(I_{p_2}-P_{22}(s)K(s))}^{-1}P_{21}(s).
\end{multline*}

An optimal controller $\bK_* = 
(A_{\bK_*},B_{\bK_*},C_{\bK_*},D_{\bK_*})$ of order $k$ then satisfies
\begin{align*}
  \bK_* =& \argmin\limits_{(\rA, \rB, \rC, \rD)}{\|P \star K\|}_\hinf, \\
         &\text{s.\,t. }
\mathcal{A}(\bK) \text{ is asymptotically stable,}
\end{align*}
where $K$ is as in \eqref{eq:controller}, the $\hinf$ norm of $P \star K$ is 
defined as
\begin{align*}
  {\|P \star K\|}_\hinf &:= \sup\limits_{s \in \C,\,\operatorname{Re}(s) > 0} 
{\|(P \star K)(s)\|}_2 \\
                        &\phantom{:}= \sup\limits_{\omega \in \R} {\|(P \star K)(\ri\omega)\|}_2,
\end{align*}
and $\mathcal{A}(\bK)$ is the closed-loop system matrix, defined as
\begin{equation} \label{eq:clmatrix}
 \begin{bmatrix} A & 0 \\ 0 & \rA \end{bmatrix} + \begin{bmatrix} B_2 & 0 
\\ 
0 & \rB \end{bmatrix} \begin{bmatrix} I_{m_2} & -\rD \\ -D_{22} & 
I_{p_2} \end{bmatrix}^{-1} \begin{bmatrix} 0 & \rC \\ C_2 & 0 
\end{bmatrix}.
\end{equation}
This matrix (resp. the closed-loop system) is called \emph{asymptotically 
stable}, if all its eigenvalues have strictly negative real part.

\section{Port-Hamiltonian Systems}\label{sec:pHsystems}
In this article, we consider plant models, for which the subsystem from controller input to measured output can be written as a pH system.
\begin{definition}\label{def:pH}
  A linear, constant-coefficient dynamical system of the form
  \begin{align}
    \label{eq:phsys}
    \begin{bmatrix}
      \dot x(t) \\ y(t)
    \end{bmatrix}
    =
    \begin{bmatrix}
      (J-R)Q & G-F \\
      {(G+F)}^\top Q & S-N
    \end{bmatrix}
    \begin{bmatrix}
      x(t) \\ u(t)
    \end{bmatrix}
  \end{align}
  is called a \emph{port-Hamiltonian system}, if the matrices $J,\, R,\, Q \in 
\R^{n \times n}$, ${G, \,F \in \R^{n \times m}}$, and ${S,\, N \in \R^{m \times 
m}}$ satisfy the following constraints:
  \begin{enumerate}[(i)]
    \item The matrices $J$ and $N$ are skew-symmetric.
    \item The \emph{passivity matrix}
      $W := \begin{bsmallmatrix}R & F \\ F^\T & S\end{bsmallmatrix}$
      is symmetric positive semi-definite.
    \item $Q$ is symmetric positive definite.
  \end{enumerate}
The \emph{Hamiltonian (energy-storage) function} $\mathcal{H}: \R^\nx  \rightarrow \R$ is then given by
\begin{align}
  \label{eq:Hamiltonian}
  \mathcal{H}(x) = \frac{1}{2}x^\T Q x.
\end{align}
\end{definition}

In this way, the plants that we consider can be formulated as
\begin{align}\label{eq:phplant}
  \begin{bmatrix}
    \dot x(t) \\ z(t) \\ y(t)
  \end{bmatrix}
  =
  \begin{bmatrix}
    (J-R)Q & B_1 & G-F \\
    C_1 & D_{11} & D_{12} \\
    {(G+F)}^\top Q & D_{21} & S-N
  \end{bmatrix}
  \begin{bmatrix}
    x(t) \\ w(t) \\ u(t)
  \end{bmatrix},
\end{align}
with matrices $J, \,R,\, Q,\, G,\, F,\, S$, and $N$ that satisfy the 
constraints imposed in \Cref{def:pH}.

Port-Hamiltonian systems are also \emph{passive}. This follows from the fact 
that the Hamiltonian~\eqref{eq:Hamiltonian} of a pH system~\eqref{eq:phsys} is 
a so-called \emph{storage function} as it fulfills the \emph{dissipation 
inequality}
\begin{equation*}
 \cH(x(t_1)) \le \cH(x(t_0)) + \int_{t_0}^{t_1} y(t)^\top u(t)\,\mathrm{d} t
\end{equation*}
for each solution trajectory $(u,x,y)$ of~\eqref{eq:phsys} and for each $t_0,\,t_1 \ge 0$ with $t_0 \le t_1$. The converse is also true under the assumption of minimality\footnote{A system is called minimal if it is controllable and observable; see~\citet{ZhoDG96}}, i.\,e., each minimal passive system can be realized as a pH system by suitable state-space transformations~\citep[Cor.~2]{BeaMX22}. For a detailed analysis of such realizations, see also~\citet{Cherifi2022}. In the next section, we use this fact to define an intuitive approach to passivity-based $\hinf$ control.

It is a well-known fact that the negative feedback interconnection of two passive systems 
is again passive (see~\citet[Theorem~6.1]{Khalil2002}) and thus \emph{Lyapunov stable}, i.\,e., all 
eigenvalues of the corresponding closed-loop system matrix have nonpositive real 
part and the ones on the imaginary axis are semi-simple. In our setting, a \emph{negative feedback 
interconnection} is a coupling of the plant and controller obtained from the modified coupling conditions $u_{\bK}(\cdot) = 
y(\cdot)$ and $u(\cdot) = -y_{\bK}(\cdot)$. We use this coupling in the remainder of  
this article.

We restate the passive interconnection result for pH systems and provide a proof  
because our setup differs slightly from the literature.

\begin{proposition}\label{prop:lyapstab}
  Consider a pH plant model (without the channels $w$ and $z$)
  \begin{align*}
    \bP_{\mathsf{pH}}:
    \begin{bmatrix}
      \dot x(t) \\ y(t)
    \end{bmatrix}
    =
    \begin{bmatrix}
      (J-R)Q & G-F \\
      (G+F)^\top Q & S-N
    \end{bmatrix}
    \begin{bmatrix}
      x(t) \\ u(t)
    \end{bmatrix},
  \end{align*}
  that satisfies the conditions in \Cref{def:pH}, and a pH controller
  \begin{align*}
    \phK: \begin{bmatrix}
      \rdx(t) \\ \red{y}(t)
    \end{bmatrix}
    =
    \begin{bmatrix}
      (\rJ-\rR)\rQ & \rG-\rF \\
      (\rG+\rF)^\top \rQ & \rS-\rN
    \end{bmatrix}
    \begin{bmatrix}
      \rx(t) \\ \red{u}(t)
    \end{bmatrix},
  \end{align*}
  that also satisfies the conditions in \Cref{def:pH}. Then the negative 
feedback interconnection of $\bP_{\mathsf{pH}}$ with $\bK_{\mathsf{pH}}$ 
resulting from the coupling conditions $u_{\bK}(\cdot) = y(\cdot)$ and $u(\cdot) 
= -y_{\bK}(\cdot)$ results in a Lyapunov stable closed-loop system.
\end{proposition}

\begin{proof}
  The closed-loop matrix of the negative feedback interconnection of 
$\bP_\mathsf{pH}$ and $\phK$ is given by
\begin{multline*} 
\small
 \begin{bmatrix} (J-R)Q & 0 \\ 0 & (J_{\bK} - R_{\bK})Q_{\bK} 
\end{bmatrix} + 
\begin{bmatrix} G-F & 0 
\\ 
0 & -G_{\bK} + F_{\bK} \end{bmatrix} \\ \small \cdot \begin{bmatrix} I_{m} 
& 
S_{\bK}-N_{\bK} \\ -S+N & 
I_{m} \end{bmatrix}^{-1} \begin{bmatrix} 0 & (G_{\bK}+F_{\bK})^\top 
Q_{\bK} 
\\ (G+F)^\top Q & 0 
\end{bmatrix}.
\end{multline*}
Note the two additional minus signs compared to~\eqref{eq:clmatrix} that are 
due to the 
negative interconnection of $\bP_\mathsf{pH}$ and $\phK$. This matrix has the 
structure of a Schur complement. By reverting the Schur complement and some 
additional permutations and scalings, it can be 
seen that this matrix has the same eigenvalues as the regular index-one matrix 
pencil (as defined in~\citet[Chapter~2]{KunM06})
\begin{align}\label{eq:phclpencil}
\footnotesize
  \begin{bmatrix}
    sI_n - (J-R)Q & 0 & -G+F & 0 \\
    0 & sI_{\nK} - (J_{\bK}-R_{\bK})Q_{\bK} & 0 & -G_{\bK} + F_{\bK} \\ 
    (G+F)^\top Q & 0 & S-N & I_{m} \\
    0 & (G_{\bK}+F_{\bK})^\top Q_{\bK} & -I_{m} & S_{\bK}-N_{\bK} 
\end{bmatrix}
\end{align}
except for $2m$ additional eigenvalues at infinity. This pencil can be written 
as $s\cE - (\cJ-\cR)\cQ$, where $\cE^\top \cQ \succeq 0$ with $\cQ = 
\operatorname{diag}(Q,Q_{\bK},I_m,I_{m}) \succ 0$, $\cJ = -\cJ^\top$, and $\cR \succeq
0$.\footnote{The notation $X \succ 0$ ($X \succeq 0$) denotes the positive 
(semi-)\\definiteness of a real symmetric matrix $X$.} Hence, it is a 
dissipative Hamiltonian pencil. Since the pencil $s\cE -\cQ$ 
is regular as well, the assertion follows from~\citet[Thm.~2]{MehMW21}.
\end{proof}

\begin{remark} \Cref{prop:lyapstab} only shows that the negative feedback 
interconnection of the plant and the controller results in a Lyapunov stable 
closed-loop system. However, the closed-loop system matrix may still have 
undesired eigenvalues on the imaginary axis. For the unstructured problem, 
there are conditions that guarantee the existence of an \emph{asymptotically} 
stabilizing full-order controller. However, for fixed-order control, these 
conditions fail to be valid, see, e.\,g., \cite{GerSS98}. Throughout our paper we assume that 
that an asymptotically stabilizing pH controller exists. An obvious sufficient 
condition that is also often satisfied in practice is the asymptotic stability 
of $(J-R)Q$; this has also been assumed in the recent work \cite{Breiten2022}. 
\end{remark}


\section{Passivity Enforcement}\label{sec:pasbased}

The close connection between passive and pH controllers (see \Cref{thm:equivPH} 
below) permits the computation of a pH controller by means of an 
alternative indirect approach: first a general controller is computed using 
either of the established fixed-order $\hinf$ synthesis methods \texttt{HIFOO} 
or \texttt{hinfstruct} and then, in a post-processing step, a passivity-check 
determines, whether the controller is already passive and can be directly 
transformed to pH form or if another passivity enforcement step must first be 
applied. This post-processing approach is typically applied successfully in 
passivity-based system identification; see e.\,g.~\cite{Gus2001, Grivet2004, 
Oliveira2014}.

If a controller $\bK$ is passive, then its transfer function $K$ is \emph{positive real} as in the following definition.
\begin{definition}\label{def:posreal}
  A proper real-rational transfer function $K$ is called \emph{positive 
real}, if
(i) all poles of $K$ have non-positive real part, 
(ii) the matrix-valued \emph{Popov function}
$
        \Phi(s) := K{(-s)}^{\T} + K(s)
        $
      attains positive semi-definite values for all $s \in \ri \R$, which are 
not poles of $K$, and
(iii) for any purely imaginary pole $\ri \omega$ of $K$ we have that the 
residue matrix $\lim \limits_{s  \rightarrow \ri \omega}(s-\ri \omega)K(s)$ is 
positive semi-definite.
\end{definition}

The equivalences between passivity, positive-real transfer functions, and a 
possible pH formulation are summarized in the following theorem. A proof can 
be found in~\cite{BeaMXZ18}.

\begin{thm}\label{thm:equivPH}
  Assume that a controller $\bK$ of order $k$ is minimal and Lyapunov stable. 
Then the following statements are equivalent:
  \begin{enumerate}[(i)]
    \item $\bK$ is passive.
    \item The controller transfer function $K$ is positive real as defined in 
\Cref{def:posreal}.
    \item $\bK$ can be written as pH controller.
    \item There exists a symmetric positive definite matrix $X \in \R^{\nK 
\times  \nK}$ satisfying the \emph{Kalman-Yakubovich-Popov (KYP) inequality}
      \begin{align*}
        \begin{bmatrix}
          -\rA^\T X- X \rA &\rC^\T -X \rB \\
          \rC - \rB^\T X & \rD + \rD^\T
        \end{bmatrix}
        \succeq 0.
      \end{align*}
  \end{enumerate}
\end{thm}

Checking $\bK$ for passivity is straight-forward. We can simply check if the KYP 
inequality can be satisfied using either an ARE or LMI solver. For 
passivity enforcement, there exist several strategies, such 
as~\cite{Grivet2004, Coelho2004, Gillis2018}. An extensive discussion on well-established 
passivity enforcement methods is presented in~\citet[Chapter 10]{Grivet2015}. In 
our numerical experiments, we use the LMI-based method for passivity enforcement 
presented in~\cite{Coelho2004} due to its moderate computational cost for small 
controller orders and its straight-forward implementation. It is based on 
computing a minimally perturbed controller output $\widetilde{C}_{\bK} := \rC + 
\Xi L_{\rm c}$, where $L_{\rm c}$ is the Cholesky factor of the 
\emph{controllability Gramian}\footnote{The controllability gramian of $\bK$ 
can be computed as solution $P_{\rm c}$ to the Lyapunov equation $\rA P_{\rm c} 
+ P_{\rm c} \rA^\T + \rB\rB^\T = 0$.} of $\bK$, such that the perturbed KYP 
inequality
\begin{align*}
  \mathcal{W}(X, \Xi) := 
  \begin{bmatrix}
    -\rA^\T X- X \rA &\widetilde{C}_{\bK}^\T -X \rB \\
    \widetilde{C}_{\bK} - \rB^\T X & \rD + \rD^\T
  \end{bmatrix}
  \succeq 0.
\end{align*}
admits a solution. The convex optimization problem that computes a minimally perturbed controller is given by
\begin{align}\label{eq:passivation}
\min {\|\Xi\|}_{\rm F} \quad \text{s. t.} \quad \mathcal{W}(X, \Xi) \succeq 0, 
\quad X \succ 0,
\end{align}
which is a standard LMI problem that we solve in our numerical experiments using the method presented in~\cite{OCPB16}.

The main problem with this passivity enforcement method is emphasized in our numerical experiments: The generic controllers computed using either 
\texttt{hinfstruct} or \texttt{HIFOO} often require a large perturbation to be 
made passive, which deteriorates the $\hinf$ performance.
Moreover, passivity enforcement methods focus on minimally changing the 
transfer function of the controller but naturally do not take the $\hinf$ 
performance of the resulting closed-loop transfer function into account. 
Instead, the method 
we present in the following section computes passive controllers that aim 
directly at minimizing the $\hinf$ norm of the closed-loop transfer function.

\section{Our Approach: Structured Optimization-Based Synthesis}\label{sec:ourmethod}

Our approach for pH $\hinf$ synthesis is an adaptation of the model 
order reduction method developed in~\cite{SchV22, SchV21}. Following a similar strategy, we only make use of samples of the closed-loop transfer function to 
avoid the computation of the $\hinf$ norm of $P \star K$ and impose the pH 
structure of the controller directly in our parameterization such that no 
constraints have to be enforced during optimization. The parameterization is 
described in \Cref{lem:param} and our optimization method is given in detail 
in \Cref{alg:sobmor}. We call our method SOBSYN (\textbf{S}tructured 
\textbf{O}ptimization-\textbf{B}ased \textbf{SYN}thesis).

\begin{lem}\label{lem:param}[\cite{SchV22}]
  Let $\theta \in \R^{n_\theta}$ be a parameter vector that is partitioned 
as
  ${\theta = {[ \theta_J^\T, \theta_W^\T, \theta_G^\T, \theta_Q^\T,\theta_N^\T 
]}^\T}$, with ${\theta_J \in \R^{\nK(\nK-1)/2}}$, ${\theta_W \in 
\R^{(\nK+p_2)(\nK+p_2+1)/2}}$, $\theta_Q \in \R^{\nK(\nK+1)/2}$, $\theta_G \in 
\R^{\nK p_2}$, and $\theta_N \in \R^{p_2(p_2-1)}$. We define the matrices
  \begin{subequations}\label{eq:param}
  \begin{align}
    \pJ(\theta) &=\vtsu{(\theta_J)}^\T-\vtsu(\theta_J),\label{eq:pJ}\\
    \pW(\theta) &=\vtu{(\theta_W)}^\T  \vtu(\theta_W),\\
    \pQ(\theta) &=\vtu{(\theta_Q)}^\T  \vtu(\theta_Q),\label{eq:pQ}\\
    \pG(\theta) &=\vtf_{k,m_2}(\theta_G), \\
    \pN(\theta) &=\vtsu{(\theta_N)}^\T - \vtsu(\theta_N),
  \end{align}
  \end{subequations}
  where the functions $\vtu: \R^{{\nK}({\nK}+1)/2}\rightarrow 
\R^{{\nK} \times {\nK}}$ [or $\vtu: \R^{({\nK+p_2})({\nK}+p_2+1)/2}\rightarrow 
\R^{({\nK}+p_2) \times ({\nK}+p_2)}$]
(resp. ${\vtsu: \R^{\nK(\nK-1)}  \rightarrow \R^{\nK \times \nK}}$) map vectors 
to upper (resp.\ strictly upper)  triangular matrices, while the function 
${\vtf: 
\R^{\nK \cdot m_2}  \rightarrow \R^{\nK \times m_2}}$ reshapes a vector of 
length $\nK m_2$ to a $\nK \times m_2$ matrix. Then, if $\theta \in \R^{n_\theta}$ is such that $\pQ(\theta)\succ 0$, a pH controller of order $k$ can be defined as
  \begin{align}
    \footnotesize
    \label{eq:phcontroller}
    \begin{bmatrix}
      \rdx(t) \\ \red{y}(t)
    \end{bmatrix}
    =
    \begin{bmatrix}
      \left( \rJ(\theta)-\rR(\theta) \right)\rQ(\theta) & 
\rG(\theta)-\rF(\theta) \\
      {(\rG(\theta)+\rF(\theta))}^\T \rQ(\theta) & \rS(\theta)-\rN(\theta)
    \end{bmatrix}
    \begin{bmatrix}
      \rx(t) \\ \red{u}(t)
    \end{bmatrix},
  \end{align}
  where $\pR(\theta), \pF(\theta)$, and $\pS(\theta)$ are extracted from $W(\theta)$ as
  \begin{align*}
      \pR(\theta) &:= \begin{bmatrix} I_{\nK} & 0 \end{bmatrix} \pW(\theta) \begin{bmatrix} I_{\nK} & 0 \end{bmatrix}^\T,\\
      \pF(\theta) &:= \begin{bmatrix} I_{\nK} & 0 \end{bmatrix} \pW(\theta) \begin{bmatrix} 0 & I_{m_2} \end{bmatrix}^\T,\\
      \pS(\theta) &:= \begin{bmatrix} 0 & I_{m_2} \end{bmatrix} \pW(\theta) \begin{bmatrix} 0 & I_{m_2} \end{bmatrix}^\T.
  \end{align*}
  Conversely, to each pH controller $\phK$ with $k$ states and $m_2$ inputs and 
outputs, a vector $\theta \in \R^{n_\theta}$ can be defined such that $\phK = 
\phK(\theta)$.
\end{lem}

This parameterization is well-suited for low-order pH systems. For higher orders, the number of parameters can be excessive. A generalization of this parameterization to pH systems with algebraic constraints is developed in~\cite{MosSMV2022b}.

We can now define the optimization problem. If $K(\theta)(\cdot)$ denotes the transfer function of 
$\phK(\theta)$, our objective function is defined as
\begin{multline}
    \loss\left(\gamma,P,K(\theta),\mathcal{S}\right) := \\
    {\frac{1}{\gamma}\sum\limits_{s_i\in 
\mathcal{S}}\sum\limits_{j=1}^{\min(m_1, p_1)}{\left({\left[\sigma_j((P \star 
K(\theta))(s_i))-\gamma\right]}_+\right)}^2,}
    \label{eq:loss}
\end{multline}
where 
\begin{align*}
  [ \cdot ]_+:  \R \rightarrow \overline{\R^+}, \quad x \mapsto 
  \begin{cases}
    x & \text{if } x\ge 0,\\
    0 & \text{if } x<0,
  \end{cases}
\end{align*}
$\mathcal{S} \in \ri \R$ is the set of sample points at which the closed-loop 
transfer function is evaluated, and $\sigma_j$ denotes the $j$-th singular value 
of its matrix argument. In \Cref{alg:sobmor}, we minimize $\loss$ w.\,r.\,t. 
$\theta$ for decreasing values of the threshold value $\gamma$ to attain a controller 
$\phK(\theta)$ that leads to a small $\hinf$ norm of the closed-loop transfer 
function. In~\citet[Remark~3.1]{SchV22} we explain the benefits resulting from a 
minimization of~$\loss$ compared to a direct minimization of the $\hinf$ norm in 
the context of model order reduction. The arguments carry over directly to 
$\hinf$ synthesis.
\begin{enumerate}
  \item The repeated computation of the $\hinf$ norm of the (potentially 
large-scale) closed-loop transfer function is computationally demanding despite 
novel approaches for the $\hinf$ norm computation of large-scale transfer 
functions~\citep{Aliyev2017, GugGO13, MitO16}. In contrast to that, an 
evaluation of $\loss$ can be carried out even for large-scale plants and the 
plant transfer function evaluations can be cached for subsequent evaluations of 
$\loss$.
  \item The $\hinf$ norm depends only continuously on the controller parameters 
    but is not differentiable in general, which requires the use of nonsmooth optimization solvers. This poses additional theoretical and computational challenges.
In contrast, $\loss$ is differentiable with respect to the controller 
parameters; see~\citet[Proposition~3.1]{SchV22}.
  \item The gradient of the $\hinf$ norm (if it exists) contains only limited 
information for an overall good descent direction because it only considers the 
closed-loop transfer function at a single point. In contrast, $\loss$ 
takes into account $P \star K(\theta)$ at all points $s_i$, where a 
singular value of $P \star K(\theta)$ is larger than $\gamma$.
\end{enumerate}

Since $\loss$ only takes into account the closed-loop transfer function at 
sample points $s_i$, it is important to choose a suitable sample point set 
$\mathcal{S}$. If the sample points are distributed such that some local maxima 
of 
${\|(P \star K(\theta))(\ri\cdot)\|}_2$ are missed entirely, then a 
minimization of $\loss$ does not lead to a good $\hinf$ performance. However, an 
abundance of sample points leads to unnecessary computational effort. 
Therefore, in~\cite{SchV21}, an adaptive sampling strategy introduced 
in~\cite{Apkarian2018} is adapted for use in combination with $\loss$. The 
method is motivated and described in detail in~\citet[Algorithm~1]{SchV21}. It 
can be used directly in our proposed $\hinf$ synthesis algorithm, which is 
described in \Cref{alg:sobmor}.

\begin{algorithm}[tb]
  \LinesNumbered
  \SetAlgoLined
  \DontPrintSemicolon
  \SetKwInOut{Input}{Input}\SetKwInOut{Output}{Output}
  \Input{Plant transfer function $P$,
  initial controller transfer function $K(\theta_0)$ with parameter 
$\theta_0 \in \R^{n_\theta}$,
  initial sample point set $\mathcal{S} \subset \mathrm{i}\R$,
  upper bound $\gamma_{\rm u} > 0$,
  bisection tolerance $\varepsilon_1 > 0$,
  termination tolerance $\varepsilon_2 > 0$
}
  \Output{$\hinf$ controller of order $k$}
  Set $j:=0$ and $\gamma_{\rm l}:=0$.\;
  \While{$(\gamma_{\rm u}-\gamma_{\rm l})/(\gamma_{\rm u}+\gamma_{\rm l}) > \varepsilon_1$}{
    Set $\gamma:=(\gamma_{\rm u}+\gamma_{\rm l})/2$.\;
    Update the sample point set $\mathcal{S}$ using~\citet[Alg.~3.1]{SchV21}. \;
    Solve the minimization problem $\alpha := \min_{\theta\in
\R^{n_\theta}}\loss(\gamma,P,K(\theta),\mathcal{S})$ with minimizer 
$\theta_{j+1} \in \R^{n_\theta}$, initialized at $\theta_j$. \;
  \eIf{$\alpha > \varepsilon_2$}{
    Set $\gamma_{\rm l}:=\gamma$.\;
    }{
    Set $\gamma_{\rm u}:=\gamma$.\;
  }
  Set $j:=j+1$.
  }
  Construct the controller with $\theta_j$ as in Lemma~\ref{lem:param}.
  \caption{SOBSYN}\label{alg:sobmor}
\end{algorithm}

SOBSYN is based on a bisection over $\gamma$. After an update of the sample 
set $\mathcal{S}$ in each iteration, we compute a minimizer for $\loss$ at the 
current $\gamma$-level using the BFGS-based optimization solver implemented 
in~\cite{mogensen2018optim}. If the minimum $\alpha$ is lower than 
$\varepsilon_2$, the solver has managed to reduce all singular values of the closed-loop transfer function evaluated at all sample points
below the current $\gamma$-level up to the \emph{termination 
tolerance} $\varepsilon_2$. In this case, we reduce the upper bound 
to~$\gamma_{\rm 
u} := \gamma$; otherwise, we increase the lower bound to~$\gamma_{\rm l} := 
\gamma$. We terminate the bisection, when the relative difference between 
$\gamma_{\rm u}$ and $\gamma_{\rm l}$ is lower than the \emph{bisection 
tolerance} $\varepsilon_1$.

\section{Numerical Experiments}\label{sec:numerics}

In our first experiment, we demonstrate the effectiveness of our method on 
the passive systems within the \COMPlib benchmark 
collection~\citep{Leibfritz2006},  
which is often used to benchmark $\hinf$ controller synthesis algorithms; 
see, e.\,g.~\cite{Apkarian2006Hinfstruct, Burke2006Hifoo}. We only use plants, 
which can be rewritten as in~\eqref{eq:phplant}, as our algorithm only applies 
to these cases. We compare our method to the general purpose fixed-order 
$\hinf$ synthesis methods \texttt{HIFOO} and \texttt{hinfstruct}.

The structure of this experiment is as follows. For each benchmark plant, we 
apply \texttt{HIFOO}, \texttt{hinfstruct}, and our method to generate a 
controller for controller orders $k = 1,\, \dots,\, 5$. After a comparison of 
the $\hinf$ performance of all controllers, we solve~\eqref{eq:passivation} to 
compute passive approximations to the controllers computed with~\texttt{HIFOO} 
and~\texttt{hinfstruct} and re-evaluate the $\hinf$ performance for the passive 
controllers. We only report a subset of the conducted experiments that is 
representative for our overall observations.\footnote{The remaining 
experimental results are available at \url{zenodo.org/record/7049952}.}

In Figure~\ref{fig:ebresults}, we plot the $\hinf$ performance of the controllers computed with our method and with~\texttt{HIFOO} and~\texttt{hinfstruct} for different models and controller orders. There are some differences in the $\hinf$ performances accross the models and controller orders but no method is systematically leading to worse results. This is despite the fact that SOBSYN is restricted to compute pH controllers only. Just for the \texttt{DLR1} model, our method leads to slightly worse controllers than the other two methods for the orders $k = 2, \dots, 5$.

\begin{figure}[tb]
  \centering
    \begin{tabular}{c}

\begin{tikzpicture}[/tikz/background rectangle/.style={fill={rgb,1:red,1.0;green,1.0;blue,1.0}, draw opacity={1.0}}, show background rectangle]
\begin{axis}[point meta max={nan}, point meta min={nan}, legend cell 
align={left}, legend columns={3}, title={}, title style={at={{(0.5,1)}}, 
anchor={south}, font={{\fontsize{14 pt}{18.2 pt}\selectfont}}, 
color={rgb,1:red,0.0;green,0.0;blue,0.0}, draw opacity={1.0}, rotate={0.0}}, 
legend style={color={rgb,1:red,0.0;green,0.0;blue,0.0}, draw opacity={1.0}, line 
width={1}, solid, fill={rgb,1:red,1.0;green,1.0;blue,1.0}, fill opacity={1.0}, 
text opacity={1.0}, font={{\fontsize{8 pt}{10.4 pt}\selectfont}}, 
text={rgb,1:red,0.0;green,0.0;blue,0.0}, cells={anchor={center}}, at={(0.5, 
1.02)}, anchor={south}}, axis 
background/.style={fill={rgb,1:red,1.0;green,1.0;blue,1.0}, opacity={1.0}}, 
anchor={north west}, xshift={1.0mm}, yshift={-1.0mm}, width={0.36\textwidth}, 
height={\hinfheight}, scaled x ticks={false}, xlabel={controller order $k$}, x 
tick style={color={rgb,1:red,0.0;green,0.0;blue,0.0}, opacity={1.0}}, x tick 
label style={color={rgb,1:red,0.0;green,0.0;blue,0.0}, opacity={1.0}, 
rotate={0}}, xlabel style={}, xmajorgrids={true}, xmin={0.88}, xmax={5.115}, 
xtick={{1.0,2.0,3.0,4.0,5.0}}, xticklabels={{$1$,$2$,$3$,$4$,$5$}}, xtick 
align={inside}, xticklabel style={font={{\fontsize{8 pt}{10.4 pt}\selectfont}}, 
color={rgb,1:red,0.0;green,0.0;blue,0.0}, draw opacity={1.0}, rotate={0.0}}, x 
grid style={color={rgb,1:red,0.0;green,0.0;blue,0.0}, draw opacity={0.1}, line 
width={0.5}, solid}, axis x line*={left}, x axis line 
style={color={rgb,1:red,0.0;green,0.0;blue,0.0}, draw opacity={1.0}, line 
width={1}, solid}, scaled y ticks={false}, ylabel={${\|P \star 
K\|}_{\mathcal{H}_\infty}$}, y tick 
style={color={rgb,1:red,0.0;green,0.0;blue,0.0}, opacity={1.0}}, y tick label 
style={color={rgb,1:red,0.0;green,0.0;blue,0.0}, opacity={1.0}, rotate={0}}, 
ylabel style={}, ymajorgrids={true}, ymin={3.103661684557664}, 
ymax={3.114261880835373}, ytick={{3.105,3.1075,3.11,3.1125000000000003}}, 
yticklabels={{$3.1050$,$3.1075$,$3.1100$,$3.1125$}}, ytick align={inside}, 
yticklabel style={font={{\fontsize{8 pt}{10.4 pt}\selectfont}}, 
color={rgb,1:red,0.0;green,0.0;blue,0.0}, draw opacity={1.0}, rotate={0.0}}, y 
grid style={color={rgb,1:red,0.0;green,0.0;blue,0.0}, draw opacity={0.1}, line 
width={0.5}, solid}, axis y line*={left}, y axis line 
style={color={rgb,1:red,0.0;green,0.0;blue,0.0}, draw opacity={1.0}, line 
width={1}, solid}, colorbar={false}]
    \addplot[color={rgb,1:red,0.302;green,0.6863;blue,0.2902}, name path={ce0ad4d5-2653-4b3b-a9ea-f25b44759a76}, draw opacity={1.0}, line width={1}, solid, mark={*}, mark size={3.0 pt}, mark repeat={1}, mark options={color={rgb,1:red,0.0;green,0.0;blue,0.0}, draw opacity={1.0}, fill={rgb,1:red,0.302;green,0.6863;blue,0.2902}, fill opacity={1.0}, line width={0.75}, rotate={0}, solid}]
        table[row sep={\\}]
        {
            \\
            1.0  3.105355661018425  \\
            2.0  3.1042005909468804  \\
            3.0  3.104230273766759  \\
            4.0  3.1042488401557486  \\
            5.0  3.104262205951297  \\
        }
        ;
    \addlegendentry {SOBSYN$ $}
    \addplot[color={rgb,1:red,0.2157;green,0.4941;blue,0.7216}, name path={f85f0d96-80ea-45d9-bb8f-9d0d5cee7390}, draw opacity={1.0}, line width={1}, solid, mark={diamond*}, mark size={3.0 pt}, mark repeat={1}, mark options={color={rgb,1:red,0.0;green,0.0;blue,0.0}, draw opacity={1.0}, fill={rgb,1:red,0.2157;green,0.4941;blue,0.7216}, fill opacity={1.0}, line width={0.75}, rotate={0}, solid}]
        table[row sep={\\}]
        {
            \\
            1.0  3.1139618752803435  \\
            2.0  3.104078730697971  \\
            3.0  3.104078729724103  \\
            4.0  3.1040787300441504  \\
            5.0  3.1040787300465342  \\
        }
        ;
    \addlegendentry {HIFOO$ $}
    \addplot[color={rgb,1:red,0.8941;green,0.102;blue,0.1098}, name path={7849d601-de71-492d-996b-0845d7cf4300}, draw opacity={1.0}, line width={1}, solid, mark={square*}, mark size={3.0 pt}, mark repeat={1}, mark options={color={rgb,1:red,0.0;green,0.0;blue,0.0}, draw opacity={1.0}, fill={rgb,1:red,0.8941;green,0.102;blue,0.1098}, fill opacity={1.0}, line width={0.75}, rotate={0}, solid}]
        table[row sep={\\}]
        {
            \\
            1.0  3.105259768487828  \\
            2.0  3.1039636487508053  \\
            3.0  3.1043645106213207  \\
            4.0  3.1039616901126936  \\
            5.0  3.1040764479189624  \\
        }
        ;
    \addlegendentry {Hinfstruct$ $}
\end{axis}
\end{tikzpicture} \\
    (a) \texttt{EB1} model \\

\begin{tikzpicture}[/tikz/background rectangle/.style={fill={rgb,1:red,1.0;green,1.0;blue,1.0}, draw opacity={1.0}}, show background rectangle]
\begin{axis}[point meta max={nan}, point meta min={nan}, legend cell 
align={left}, legend columns={3}, title={}, title style={at={{(0.5,1)}}, 
anchor={south}, font={{\fontsize{14 pt}{18.2 pt}\selectfont}}, 
color={rgb,1:red,0.0;green,0.0;blue,0.0}, draw opacity={1.0}, rotate={0.0}}, 
legend style={color={rgb,1:red,0.0;green,0.0;blue,0.0}, draw opacity={1.0}, line 
width={1}, solid, fill={rgb,1:red,1.0;green,1.0;blue,1.0}, fill opacity={1.0}, 
text opacity={1.0}, font={{\fontsize{8 pt}{10.4 pt}\selectfont}}, 
text={rgb,1:red,0.0;green,0.0;blue,0.0}, cells={anchor={center}}, at={(1.02, 
1)}, anchor={north west}}, axis 
background/.style={fill={rgb,1:red,1.0;green,1.0;blue,1.0}, opacity={1.0}}, 
anchor={north west}, xshift={1.0mm}, yshift={-1.0mm}, width={0.36\textwidth}, 
height={\hinfheight}, scaled x ticks={false}, xlabel={controller order $k$}, x 
tick style={color={rgb,1:red,0.0;green,0.0;blue,0.0}, opacity={1.0}}, x tick 
label style={color={rgb,1:red,0.0;green,0.0;blue,0.0}, opacity={1.0}, 
rotate={0}}, xlabel style={}, xmajorgrids={true}, xmin={0.88}, xmax={5.12}, 
xtick={{1.0,2.0,3.0,4.0,5.0}}, xticklabels={{$1$,$2$,$3$,$4$,$5$}}, xtick 
align={inside}, xticklabel style={font={{\fontsize{8 pt}{10.4 pt}\selectfont}}, 
color={rgb,1:red,0.0;green,0.0;blue,0.0}, draw opacity={1.0}, rotate={0.0}}, x 
grid style={color={rgb,1:red,0.0;green,0.0;blue,0.0}, draw opacity={0.1}, line 
width={0.5}, solid}, axis x line*={left}, x axis line 
style={color={rgb,1:red,0.0;green,0.0;blue,0.0}, draw opacity={1.0}, line 
width={1}, solid}, scaled y ticks={false}, ylabel={${\|P \star 
K\|}_{\mathcal{H}_\infty}$}, y tick 
style={color={rgb,1:red,0.0;green,0.0;blue,0.0}, opacity={1.0}}, y tick label 
style={color={rgb,1:red,0.0;green,0.0;blue,0.0}, opacity={1.0}, rotate={0}}, 
ylabel style={}, ymajorgrids={true}, ymin={1.8062114550288484}, 
ymax={1.9082132592972785}, 
ytick={{1.82,1.84,1.86,1.8800000000000001,1.9000000000000001}}, 
yticklabels={{$1.82$,$1.84$,$1.86$,$1.88$,$1.90$}}, ytick align={inside}, 
yticklabel style={font={{\fontsize{8 pt}{10.4 pt}\selectfont}}, 
color={rgb,1:red,0.0;green,0.0;blue,0.0}, draw opacity={1.0}, rotate={0.0}}, y 
grid style={color={rgb,1:red,0.0;green,0.0;blue,0.0}, draw opacity={0.1}, line 
width={0.5}, solid}, axis y line*={left}, y axis line 
style={color={rgb,1:red,0.0;green,0.0;blue,0.0}, draw opacity={1.0}, line 
width={1}, solid}, colorbar={false}]
    \addplot[color={rgb,1:red,0.302;green,0.6863;blue,0.2902}, name path={a36947d5-d9bd-4e14-b535-d042077801a1}, draw opacity={1.0}, line width={1}, solid, mark={*}, mark size={3.0 pt}, mark repeat={1}, mark options={color={rgb,1:red,0.0;green,0.0;blue,0.0}, draw opacity={1.0}, fill={rgb,1:red,0.302;green,0.6863;blue,0.2902}, fill opacity={1.0}, line width={0.75}, rotate={0}, solid}]
        table[row sep={\\}]
        {
            \\
            1.0  1.9053264157802474  \\
            2.0  1.8645983206452803  \\
            3.0  1.8292223552137954  \\
            4.0  1.8161573906278334  \\
            5.0  1.8129812728419485  \\
        }
        ;
    \addplot[color={rgb,1:red,0.2157;green,0.4941;blue,0.7216}, name path={4dd92abc-3724-45f1-a63a-d03a5ae0aa4f}, draw opacity={1.0}, line width={1}, solid, mark={diamond*}, mark size={3.0 pt}, mark repeat={1}, mark options={color={rgb,1:red,0.0;green,0.0;blue,0.0}, draw opacity={1.0}, fill={rgb,1:red,0.2157;green,0.4941;blue,0.7216}, fill opacity={1.0}, line width={0.75}, rotate={0}, solid}]
        table[row sep={\\}]
        {
            \\
            1.0  1.9052356587696306  \\
            2.0  1.8645500368107046  \\
            3.0  1.8184723752897263  \\
            4.0  1.8116628501877088  \\
            5.0  1.8090982985458794  \\
        }
        ;
    \addplot[color={rgb,1:red,0.8941;green,0.102;blue,0.1098}, name path={ed3c3ed0-c32d-4156-b111-2f1a9167eba4}, draw opacity={1.0}, line width={1}, solid, mark={square*}, mark size={3.0 pt}, mark repeat={1}, mark options={color={rgb,1:red,0.0;green,0.0;blue,0.0}, draw opacity={1.0}, fill={rgb,1:red,0.8941;green,0.102;blue,0.1098}, fill opacity={1.0}, line width={0.75}, rotate={0}, solid}]
        table[row sep={\\}]
        {
            \\
            1.0  1.905235732184757  \\
            2.0  1.9030652261583354  \\
            3.0  1.8180161030991921  \\
            4.0  1.8182168260010907  \\
            5.0  1.8184266274386391  \\
        }
        ;
\end{axis}
\end{tikzpicture} \\
    (b) \texttt{EB6} model \\

\begin{tikzpicture}[/tikz/background rectangle/.style={fill={rgb,1:red,1.0;green,1.0;blue,1.0}, draw opacity={1.0}}, show background rectangle]
\begin{axis}[point meta max={nan}, point meta min={nan}, legend cell 
align={left}, legend columns={3}, title={}, title style={at={{(0.5,1)}}, 
anchor={south}, font={{\fontsize{14 pt}{18.2 pt}\selectfont}}, 
color={rgb,1:red,0.0;green,0.0;blue,0.0}, draw opacity={1.0}, rotate={0.0}}, 
legend style={color={rgb,1:red,0.0;green,0.0;blue,0.0}, draw opacity={1.0}, line 
width={1}, solid, fill={rgb,1:red,1.0;green,1.0;blue,1.0}, fill opacity={1.0}, 
text opacity={1.0}, font={{\fontsize{8 pt}{10.4 pt}\selectfont}}, 
text={rgb,1:red,0.0;green,0.0;blue,0.0}, cells={anchor={center}}, at={(1.02, 
1)}, anchor={north west}}, axis 
background/.style={fill={rgb,1:red,1.0;green,1.0;blue,1.0}, opacity={1.0}}, 
anchor={north west}, xshift={1.0mm}, yshift={-1.0mm}, width={0.36\textwidth}, 
height={\hinfheight}, scaled x ticks={false}, xlabel={controller order $k$}, x 
tick style={color={rgb,1:red,0.0;green,0.0;blue,0.0}, opacity={1.0}}, x tick 
label style={color={rgb,1:red,0.0;green,0.0;blue,0.0}, opacity={1.0}, 
rotate={0}}, xlabel style={}, xmajorgrids={true}, xmin={0.88}, xmax={5.12}, 
xtick={{1.0,2.0,3.0,4.0,5.0}}, xticklabels={{$1$,$2$,$3$,$4$,$5$}}, xtick 
align={inside}, xticklabel style={font={{\fontsize{8 pt}{10.4 pt}\selectfont}}, 
color={rgb,1:red,0.0;green,0.0;blue,0.0}, draw opacity={1.0}, rotate={0.0}}, x 
grid style={color={rgb,1:red,0.0;green,0.0;blue,0.0}, draw opacity={0.1}, line 
width={0.5}, solid}, axis x line*={left}, x axis line 
style={color={rgb,1:red,0.0;green,0.0;blue,0.0}, draw opacity={1.0}, line 
width={1}, solid}, scaled y ticks={false}, ylabel={${\|P \star 
K\|}_{\mathcal{H}_\infty}$}, y tick 
style={color={rgb,1:red,0.0;green,0.0;blue,0.0}, opacity={1.0}}, y tick label 
style={color={rgb,1:red,0.0;green,0.0;blue,0.0}, opacity={1.0}, rotate={0}}, 
ylabel style={}, ymajorgrids={true}, ymin={-0.01041290431797455}, 
ymax={2.8568207621502397}, ytick={{0.0,0.5,1.0,1.5,2.0,2.5}}, 
yticklabels={{$0.0$,$0.5$,$1.0$,$1.5$,$2.0$,$2.5$}}, ytick align={inside}, 
yticklabel style={font={{\fontsize{8 pt}{10.4 pt}\selectfont}}, 
color={rgb,1:red,0.0;green,0.0;blue,0.0}, draw opacity={1.0}, rotate={0.0}}, y 
grid style={color={rgb,1:red,0.0;green,0.0;blue,0.0}, draw opacity={0.1}, line 
width={0.5}, solid}, axis y line*={left}, y axis line 
style={color={rgb,1:red,0.0;green,0.0;blue,0.0}, draw opacity={1.0}, line 
width={1}, solid}, colorbar={false}]
    \addplot[color={rgb,1:red,0.302;green,0.6863;blue,0.2902}, name path={16607f74-26de-47ce-ad4d-df38ed3945c0}, draw opacity={1.0}, line width={1}, solid, mark={*}, mark size={3.0 pt}, mark repeat={1}, mark options={color={rgb,1:red,0.0;green,0.0;blue,0.0}, draw opacity={1.0}, fill={rgb,1:red,0.302;green,0.6863;blue,0.2902}, fill opacity={1.0}, line width={0.75}, rotate={0}, solid}]
        table[row sep={\\}]
        {
            \\
            1.0  1.365766195555964  \\
            2.0  0.5096737304377291  \\
            3.0  0.46816816931818905  \\
            4.0  0.4396580971241398  \\
            5.0  0.38125921387967837  \\
        }
        ;
    \addplot[color={rgb,1:red,0.2157;green,0.4941;blue,0.7216}, name path={7d0e5c5e-8387-4eb7-b030-8aa7e61326a6}, draw opacity={1.0}, line width={1}, solid, mark={diamond*}, mark size={3.0 pt}, mark repeat={1}, mark options={color={rgb,1:red,0.0;green,0.0;blue,0.0}, draw opacity={1.0}, fill={rgb,1:red,0.2157;green,0.4941;blue,0.7216}, fill opacity={1.0}, line width={0.75}, rotate={0}, solid}]
        table[row sep={\\}]
        {
            \\
            1.0  0.6666420523805608  \\
            2.0  0.24297515899234523  \\
            3.0  0.20551917445888132  \\
            4.0  0.07412706202328107  \\
            5.0  0.0707352183179183  \\
        }
        ;
    \addplot[color={rgb,1:red,0.8941;green,0.102;blue,0.1098}, name path={0df2ee71-4ea3-401e-98ec-e0df1846b7ea}, draw opacity={1.0}, line width={1}, solid, mark={square*}, mark size={3.0 pt}, mark repeat={1}, mark options={color={rgb,1:red,0.0;green,0.0;blue,0.0}, draw opacity={1.0}, fill={rgb,1:red,0.8941;green,0.102;blue,0.1098}, fill opacity={1.0}, line width={0.75}, rotate={0}, solid}]
        table[row sep={\\}]
        {
            \\
            1.0  2.775672639514347  \\
            2.0  0.24297938875991948  \\
            3.0  0.220800249923581  \\
            4.0  0.07563577844871887  \\
            5.0  0.14270439766798199  \\
        }
        ;
\end{axis}
\end{tikzpicture} \\
    (c) \texttt{DLR1} model
    \end{tabular}
  \caption{$\hinf$ performance for general controllers}\label{fig:ebresults}
\end{figure}

In Figure~\ref{fig:ebresults_passive}, we compare the $\hinf$ performance of 
\mbox{SOBSYN} to the $\hinf$ performance of the \texttt{HIFOO} and 
\texttt{hinfstruct} controllers after the passivation step is applied. For 
\texttt{EB1}, the \texttt{HIFOO} and \texttt{hinfstruct} controllers are already mostly 
passive such that only for $r \in \{3, 5\}$, the $\hinf$ performance is 
slightly worse for \texttt{hinfstruct}\@. In \texttt{EB6}, we can observe the 
potentially damaging behavior of the passivation step to the $\hinf$ 
performance. Here the $\hinf$ norm of the closed-loop transfer function 
increases by more than two orders of magnitudes after the passivation. In 
Figure~\ref{fig:ebresults_passive} (c) we can observe, that for the \texttt{DLR1} model, 
most \texttt{HIFOO} and \texttt{hinfstruct} controllers are not passive such 
that the passivation step deteriorates their $\hinf$ performance significantly. 
In Figure~\ref{fig:popov}, we show the eigenvalues of the Popov function 
evaluated on the imaginary axis for one of the controllers computed with 
\texttt{HIFOO}. The figure reveals that some the eigenvalues are significantly below zero, which 
emphasizes that passivity enforcement often has to change the general purpose controllers drastically to make them passive.

\begin{figure}[tb]
  \centering
  \input{./PlotSources/popovplot.tikz}
  \caption{Eigenvalues of the Popov function of the order 5 controller for 
the \texttt{DLR1} model computed with \texttt{HIFOO}}\label{fig:popov}
\end{figure}

\begin{figure}[tb]
  \centering
    \begin{tabular}{c}

\begin{tikzpicture}[/tikz/background rectangle/.style={fill={rgb,1:red,1.0;green,1.0;blue,1.0}, draw opacity={1.0}}, show background rectangle]
\begin{axis}[point meta max={nan}, point meta min={nan}, legend cell 
align={left}, legend columns={3}, title={}, title style={at={{(0.5,1)}}, 
anchor={south}, font={{\fontsize{14 pt}{18.2 pt}\selectfont}}, 
color={rgb,1:red,0.0;green,0.0;blue,0.0}, draw opacity={1.0}, rotate={0.0}}, 
legend style={color={rgb,1:red,0.0;green,0.0;blue,0.0}, draw opacity={1.0}, line 
width={1}, solid, fill={rgb,1:red,1.0;green,1.0;blue,1.0}, fill opacity={1.0}, 
text opacity={1.0}, font={{\fontsize{8 pt}{10.4 pt}\selectfont}}, 
text={rgb,1:red,0.0;green,0.0;blue,0.0}, cells={anchor={center}}, at={(0.5, 
1.02)}, anchor={south}}, axis 
background/.style={fill={rgb,1:red,1.0;green,1.0;blue,1.0}, opacity={1.0}}, 
anchor={north west}, xshift={1.0mm}, yshift={-1.0mm}, width={0.36\textwidth}, 
height={\hinfheight}, scaled x ticks={false}, xlabel={controller order $k$}, x 
tick style={color={rgb,1:red,0.0;green,0.0;blue,0.0}, opacity={1.0}}, x tick 
label style={color={rgb,1:red,0.0;green,0.0;blue,0.0}, opacity={1.0}, 
rotate={0}}, xlabel style={}, xmajorgrids={true}, xmin={0.88}, xmax={5.12}, 
xtick={{1.0,2.0,3.0,4.0,5.0}}, xticklabels={{$1$,$2$,$3$,$4$,$5$}}, xtick 
align={inside}, xticklabel style={font={{\fontsize{8 pt}{10.4 pt}\selectfont}}, 
color={rgb,1:red,0.0;green,0.0;blue,0.0}, draw opacity={1.0}, rotate={0.0}}, x 
grid style={color={rgb,1:red,0.0;green,0.0;blue,0.0}, draw opacity={0.1}, line 
width={0.5}, solid}, axis x line*={left}, x axis line 
style={color={rgb,1:red,0.0;green,0.0;blue,0.0}, draw opacity={1.0}, line 
width={1}, solid}, scaled y ticks={false}, ylabel={${\|P \star 
K\|}_{\mathcal{H}_\infty}$}, y tick 
style={color={rgb,1:red,0.0;green,0.0;blue,0.0}, opacity={1.0}}, y tick label 
style={color={rgb,1:red,0.0;green,0.0;blue,0.0}, opacity={1.0}, rotate={0}}, 
ylabel style={}, ymajorgrids={true}, ymin={3.1038845240448563}, 
ymax={3.1143563121844506}, ytick={{3.105,3.1075,3.11,3.1125000000000003}}, 
yticklabels={{$3.1050$,$3.1075$,$3.1100$,$3.1125$}}, ytick align={inside}, 
yticklabel style={font={{\fontsize{8 pt}{10.4 pt}\selectfont}}, 
color={rgb,1:red,0.0;green,0.0;blue,0.0}, draw opacity={1.0}, rotate={0.0}}, y 
grid style={color={rgb,1:red,0.0;green,0.0;blue,0.0}, draw opacity={0.1}, line 
width={0.5}, solid}, axis y line*={left}, y axis line 
style={color={rgb,1:red,0.0;green,0.0;blue,0.0}, draw opacity={1.0}, line 
width={1}, solid}, colorbar={false}]
    \addplot[color={rgb,1:red,0.302;green,0.6863;blue,0.2902}, name path={ade11ef1-41e2-476a-a1ee-3305aaae987d}, draw opacity={1.0}, line width={1}, solid, mark={*}, mark size={3.0 pt}, mark repeat={1}, mark options={color={rgb,1:red,0.0;green,0.0;blue,0.0}, draw opacity={1.0}, fill={rgb,1:red,0.302;green,0.6863;blue,0.2902}, fill opacity={1.0}, line width={0.75}, rotate={0}, solid}]
        table[row sep={\\}]
        {
            \\
            1.0  3.105355661018425  \\
            2.0  3.1042005909468804  \\
            3.0  3.104230273766759  \\
            4.0  3.1042488401557486  \\
            5.0  3.104262205951297  \\
        }
        ;
    \addlegendentry {SOBSYN$ $}
    \addplot[color={rgb,1:red,0.2157;green,0.4941;blue,0.7216}, name path={06d28c30-533b-49a2-b9d0-d2f9c2c2f222}, draw opacity={1.0}, line width={1}, dotted, mark={diamond*}, mark size={3.0 pt}, mark repeat={1}, mark options={color={rgb,1:red,0.0;green,0.0;blue,0.0}, draw opacity={1.0}, fill={rgb,1:red,0.2157;green,0.4941;blue,0.7216}, fill opacity={1.0}, line width={0.75}, rotate={0}, solid}]
        table[row sep={\\}]
        {
            \\
            1.0  3.1140599408220093  \\
            2.0  3.1041808954072976  \\
            3.0  3.104180902876261  \\
            4.0  3.1041809026656133  \\
            5.0  3.1041809016459263  \\
        }
        ;
    \addlegendentry {HIFOO-PS$ $}
    \addplot[color={rgb,1:red,0.8941;green,0.102;blue,0.1098}, name path={713b7d8e-70a3-4f9d-817a-123eeeb7c11e}, draw opacity={1.0}, line width={1}, dotted, mark={square*}, mark size={3.0 pt}, mark repeat={1}, mark options={color={rgb,1:red,0.0;green,0.0;blue,0.0}, draw opacity={1.0}, fill={rgb,1:red,0.8941;green,0.102;blue,0.1098}, fill opacity={1.0}, line width={0.75}, rotate={0}, solid}]
        table[row sep={\\}]
        {
            \\
            1.0  3.1053556810314404  \\
            2.0  3.1041809045481625  \\
            3.0  3.1065900029113696  \\
            4.0  3.104180904552181  \\
            5.0  3.1048867188143356  \\
        }
        ;
    \addlegendentry {Hinfstruct-PS$ $}
\end{axis}
\end{tikzpicture} \\
    (a) \texttt{EB1} model \\

\begin{tikzpicture}[/tikz/background rectangle/.style={fill={rgb,1:red,1.0;green,1.0;blue,1.0}, draw opacity={1.0}}, show background rectangle]
\begin{axis}[point meta max={nan}, point meta min={nan}, legend cell 
align={left}, legend columns={3}, title={}, title style={at={{(0.5,1)}}, 
anchor={south}, font={{\fontsize{14 pt}{18.2 pt}\selectfont}}, 
color={rgb,1:red,0.0;green,0.0;blue,0.0}, draw opacity={1.0}, rotate={0.0}}, 
legend style={color={rgb,1:red,0.0;green,0.0;blue,0.0}, draw opacity={1.0}, line 
width={1}, solid, fill={rgb,1:red,1.0;green,1.0;blue,1.0}, fill opacity={1.0}, 
text opacity={1.0}, font={{\fontsize{8 pt}{10.4 pt}\selectfont}}, 
text={rgb,1:red,0.0;green,0.0;blue,0.0}, cells={anchor={center}}, at={(1.02, 
1)}, anchor={north west}}, axis 
background/.style={fill={rgb,1:red,1.0;green,1.0;blue,1.0}, opacity={1.0}}, 
anchor={north west}, xshift={1.0mm}, yshift={-1.0mm}, width={0.36\textwidth}, 
height={\hinfheight}, scaled x ticks={false}, xlabel={controller order $k$}, x 
tick style={color={rgb,1:red,0.0;green,0.0;blue,0.0}, opacity={1.0}}, x tick 
label style={color={rgb,1:red,0.0;green,0.0;blue,0.0}, opacity={1.0}, 
rotate={0}}, xlabel style={}, xmajorgrids={true}, xmin={0.88}, xmax={5.12}, 
xtick={{1.0,2.0,3.0,4.0,5.0}}, xticklabels={{$1$,$2$,$3$,$4$,$5$}}, xtick 
align={inside}, xticklabel style={font={{\fontsize{8 pt}{10.4 pt}\selectfont}}, 
color={rgb,1:red,0.0;green,0.0;blue,0.0}, draw opacity={1.0}, rotate={0.0}}, x 
grid style={color={rgb,1:red,0.0;green,0.0;blue,0.0}, draw opacity={0.1}, line 
width={0.5}, solid}, axis x line*={left}, x axis line 
style={color={rgb,1:red,0.0;green,0.0;blue,0.0}, draw opacity={1.0}, line 
width={1}, solid}, scaled y ticks={false}, ylabel={${\|P \star 
K\|}_{\mathcal{H}_\infty}$}, y tick 
style={color={rgb,1:red,0.0;green,0.0;blue,0.0}, opacity={1.0}}, y tick label 
style={color={rgb,1:red,0.0;green,0.0;blue,0.0}, opacity={1.0}, rotate={0}}, 
ylabel style={}, ymajorgrids={true}, ymin={-10.591270398865015}, 
ymax={427.69228866811443}, ytick={{0.0,100.0,200.0,300.0,400.0}}, 
yticklabels={{$0$,$100$,$200$,$300$,$400$}}, ytick align={inside}, yticklabel 
style={font={{\fontsize{8 pt}{10.4 pt}\selectfont}}, 
color={rgb,1:red,0.0;green,0.0;blue,0.0}, draw opacity={1.0}, rotate={0.0}}, y 
grid style={color={rgb,1:red,0.0;green,0.0;blue,0.0}, draw opacity={0.1}, line 
width={0.5}, solid}, axis y line*={left}, y axis line 
style={color={rgb,1:red,0.0;green,0.0;blue,0.0}, draw opacity={1.0}, line 
width={1}, solid}, colorbar={false}]
    \addplot[color={rgb,1:red,0.302;green,0.6863;blue,0.2902}, name path={14d24e21-5ab6-4731-bc0d-11ab72029c76}, draw opacity={1.0}, line width={1}, solid, mark={*}, mark size={3.0 pt}, mark repeat={1}, mark options={color={rgb,1:red,0.0;green,0.0;blue,0.0}, draw opacity={1.0}, fill={rgb,1:red,0.302;green,0.6863;blue,0.2902}, fill opacity={1.0}, line width={0.75}, rotate={0}, solid}]
        table[row sep={\\}]
        {
            \\
            1.0  1.9053264157802474  \\
            2.0  1.8645983206452803  \\
            3.0  1.8292223552137954  \\
            4.0  1.8161573906278334  \\
            5.0  1.8129812728419485  \\
        }
        ;
    \addplot[color={rgb,1:red,0.2157;green,0.4941;blue,0.7216}, name path={12791cf9-1435-45ad-b928-583a8a260437}, draw opacity={1.0}, line width={1}, dotted, mark={diamond*}, mark size={3.0 pt}, mark repeat={1}, mark options={color={rgb,1:red,0.0;green,0.0;blue,0.0}, draw opacity={1.0}, fill={rgb,1:red,0.2157;green,0.4941;blue,0.7216}, fill opacity={1.0}, line width={0.75}, rotate={0}, solid}]
        table[row sep={\\}]
        {
            \\
            1.0  1.9052358432122911  \\
            2.0  1.8645500273465538  \\
            3.0  1.905891280907725  \\
            4.0  415.28803699640747  \\
            5.0  2.858626917560559  \\
        }
        ;
    \addplot[color={rgb,1:red,0.8941;green,0.102;blue,0.1098}, name path={43f4f184-e1ef-4aba-9213-3d4cca2edc61}, draw opacity={1.0}, line width={1}, dotted, mark={square*}, mark size={3.0 pt}, mark repeat={1}, mark options={color={rgb,1:red,0.0;green,0.0;blue,0.0}, draw opacity={1.0}, fill={rgb,1:red,0.8941;green,0.102;blue,0.1098}, fill opacity={1.0}, line width={0.75}, rotate={0}, solid}]
        table[row sep={\\}]
        {
            \\
            1.0  1.905235727041193  \\
            2.0  1.903065229489284  \\
            3.0  2.5221285445063315  \\
            4.0  5.12381661912028  \\
            5.0  2.1840741590416357  \\
        }
        ;
\end{axis}
\end{tikzpicture} \\
    (b) \texttt{EB6} model \\

\begin{tikzpicture}[/tikz/background rectangle/.style={fill={rgb,1:red,1.0;green,1.0;blue,1.0}, draw opacity={1.0}}, show background rectangle]
\begin{axis}[point meta max={nan}, point meta min={nan}, legend cell 
align={left}, legend columns={3}, title={}, title style={at={{(0.5,1)}}, 
anchor={south}, font={{\fontsize{14 pt}{18.2 pt}\selectfont}}, 
color={rgb,1:red,0.0;green,0.0;blue,0.0}, draw opacity={1.0}, rotate={0.0}}, 
legend style={color={rgb,1:red,0.0;green,0.0;blue,0.0}, draw opacity={1.0}, line 
width={1}, solid, fill={rgb,1:red,1.0;green,1.0;blue,1.0}, fill opacity={1.0}, 
text opacity={1.0}, font={{\fontsize{8 pt}{10.4 pt}\selectfont}}, 
text={rgb,1:red,0.0;green,0.0;blue,0.0}, cells={anchor={center}}, at={(1.02, 
1)}, anchor={north west}}, axis 
background/.style={fill={rgb,1:red,1.0;green,1.0;blue,1.0}, opacity={1.0}}, 
anchor={north west}, xshift={1.0mm}, yshift={-1.0mm}, width={0.36\textwidth}, 
height={\hinfheight}, scaled x ticks={false}, xlabel={controller order $k$}, x 
tick style={color={rgb,1:red,0.0;green,0.0;blue,0.0}, opacity={1.0}}, x tick 
label style={color={rgb,1:red,0.0;green,0.0;blue,0.0}, opacity={1.0}, 
rotate={0}}, xlabel style={}, xmajorgrids={true}, xmin={0.88}, xmax={5.12}, 
xtick={{1.0,2.0,3.0,4.0,5.0}}, xticklabels={{$1$,$2$,$3$,$4$,$5$}}, xtick 
align={inside}, xticklabel style={font={{\fontsize{8 pt}{10.4 pt}\selectfont}}, 
color={rgb,1:red,0.0;green,0.0;blue,0.0}, draw opacity={1.0}, rotate={0.0}}, x 
grid style={color={rgb,1:red,0.0;green,0.0;blue,0.0}, draw opacity={0.1}, line 
width={0.5}, solid}, axis x line*={left}, x axis line 
style={color={rgb,1:red,0.0;green,0.0;blue,0.0}, draw opacity={1.0}, line 
width={1}, solid}, scaled y ticks={false}, ylabel={${\|P \star 
K\|}_{\mathcal{H}_\infty}$}, y tick 
style={color={rgb,1:red,0.0;green,0.0;blue,0.0}, opacity={1.0}}, y tick label 
style={color={rgb,1:red,0.0;green,0.0;blue,0.0}, opacity={1.0}, rotate={0}}, 
ylabel style={}, ymajorgrids={true}, ymin={-0.059609830315016016}, 
ymax={7.089580607567285}, ytick={{0.0,2.0,4.0,6.0}}, 
yticklabels={{$0$,$2$,$4$,$6$}}, ytick align={inside}, yticklabel 
style={font={{\fontsize{8 pt}{10.4 pt}\selectfont}}, 
color={rgb,1:red,0.0;green,0.0;blue,0.0}, draw opacity={1.0}, rotate={0.0}}, y 
grid style={color={rgb,1:red,0.0;green,0.0;blue,0.0}, draw opacity={0.1}, line 
width={0.5}, solid}, axis y line*={left}, y axis line 
style={color={rgb,1:red,0.0;green,0.0;blue,0.0}, draw opacity={1.0}, line 
width={1}, solid}, colorbar={false}]
    \addplot[color={rgb,1:red,0.302;green,0.6863;blue,0.2902}, name path={af24e97c-5fd2-4707-84f4-bf86af83a3af}, draw opacity={1.0}, line width={1}, solid, mark={*}, mark size={3.0 pt}, mark repeat={1}, mark options={color={rgb,1:red,0.0;green,0.0;blue,0.0}, draw opacity={1.0}, fill={rgb,1:red,0.302;green,0.6863;blue,0.2902}, fill opacity={1.0}, line width={0.75}, rotate={0}, solid}]
        table[row sep={\\}]
        {
            \\
            1.0  1.365766195555964  \\
            2.0  0.5096737304377291  \\
            3.0  0.46816816931818905  \\
            4.0  0.4396580971241398  \\
            5.0  0.38125921387967837  \\
        }
        ;
    \addplot[color={rgb,1:red,0.2157;green,0.4941;blue,0.7216}, name path={210ce619-4df3-49e6-90f0-e338e1ea82bf}, draw opacity={1.0}, line width={1}, dotted, mark={diamond*}, mark size={3.0 pt}, mark repeat={1}, mark options={color={rgb,1:red,0.0;green,0.0;blue,0.0}, draw opacity={1.0}, fill={rgb,1:red,0.2157;green,0.4941;blue,0.7216}, fill opacity={1.0}, line width={0.75}, rotate={0}, solid}]
        table[row sep={\\}]
        {
            \\
            1.0  0.6666420523804074  \\
            2.0  6.388618484662486  \\
            3.0  4.404174318833524  \\
            4.0  6.887245029136654  \\
            5.0  5.221258417771321  \\
        }
        ;
    \addplot[color={rgb,1:red,0.8941;green,0.102;blue,0.1098}, name path={9d9e1315-b767-4a27-a417-5d3b5d196060}, draw opacity={1.0}, line width={1}, dotted, mark={square*}, mark size={3.0 pt}, mark repeat={1}, mark options={color={rgb,1:red,0.0;green,0.0;blue,0.0}, draw opacity={1.0}, fill={rgb,1:red,0.8941;green,0.102;blue,0.1098}, fill opacity={1.0}, line width={0.75}, rotate={0}, solid}]
        table[row sep={\\}]
        {
            \\
            1.0  2.7756726425270033  \\
            2.0  5.02201430285878  \\
            3.0  6.742816103369417  \\
            4.0  4.055780585125054  \\
            5.0  0.14272574811561514  \\
        }
        ;
\end{axis}
\end{tikzpicture} \\
    (c) \texttt{DLR1} model
    \end{tabular}
  \caption{$\hinf$ performance for passive controllers}\label{fig:ebresults_passive}
\end{figure}


In our second experiment we use a scalable mass-spring-damper (MSD) system in port-Hamiltonian 
form to evaluate the impact of the plant order on the runtime of our method. The model is derived in~\cite{Gugercin2012} and the corresponding $\hinf$ 
control problem is described in the collection of pH benchmark 
systems.\footnote{available at \url{https://perma.cc/297G-NFUJ}} The 
state-space dimension of the model can be scaled by changing the 
number of masses that are included in the MSD chain. The runtimes and $\hinf$ 
performances for plant orders from 10 to 2000 and controller orders 1, 5, and 10 are reported in \Cref{tab:msdresults}. It can be observed, that 
while the runtimes increase significantly for \texttt{HIFOO} and 
\texttt{hinfstruct} as the plant dimension $n$ is increased (\texttt{hinfstruct} 
takes almost 50 hours to compute a controller for the plant with state dimension 
2000), the runtime of SOBSYN is less affected from an increase of $n$. However, 
we can observe that the controller dimension affects the runtime, which is 
well-aligned with our previous findings~\citep{SchV21, SchV22}. In general, the 
runtime of SOBSYN stays well under an hour for all plant and controller orders. 
The overall best $\hinf$ performance is provided by \texttt{hinfstruct} accross 
the different plant and controller orders. \texttt{HIFOO} often leads to a 
slightly worse performance than \texttt{hinfstruct}. In constrast to the 
previous experiment, SOBSYN has the worst $\hinf$ performance for low plant 
orders (up to 50) but for medium to large plant orders, SOBSYN achieves an 
$\hinf$ performance, that ranks between \texttt{hinfstruct} and \texttt{HIFOO}. 
When comparing the $\hinf$ performance of SOBSYN to the other two methods, it is important to note that the other they 
only result in one passive controller in this second experiment (for 
\texttt{HIFOO} at $n=1000$ and $k=5$). We can again expect a worse performance 
after a passivation step is applied as demonstrated in the previous experiment. 
However, this also emphasizes the fact that for certain plants, general 
controllers can lead to a better $\hinf$ performance than controllers that are 
restricted to be passive.

\begin{table*}[htpb]
  \centering
  \caption{$\hinf$ performance and runtimes of \texttt{HIFOO}, \texttt{hinfstruct}, and SOBSYN for the MSD plants at different state dimensions. The runtimes are given in seconds. For $n=2000$, \texttt{HIFOO} failed for all controller orders. We use an Intel\textsuperscript{\textregistered}\,Core\texttrademark\; i9-9900K CPU at 3.60\,GHz with 32\,GB of RAM.}
  \label{tab:msdresults}
  \begin{tabular}{cc|ccc|ccc|ccc}
    &&\multicolumn{3}{c}{\texttt{HIFOO}} & \multicolumn{3}{c}{\texttt{hinstruct}} & \multicolumn{3}{c}{SOBSYN} \\ 
    $n$&& $k=1$ & $k=5$ & $k=10$ & $k=1$ & $k=5$ & $k=10$ &$k=1$ & $k=5$ & $k=10$ \\ \hline
    \multirow{2}{*}{$10$}   & $\hinf$-norm & 5.2e$-$01 & 3.8e$-$01 & 4.4e$-$01 & 4.3e$-$01 & 4.2e$-$01 & 5.1e$-$01 & 4.9e$-$01 & 4.6e$-$01 & 4.6e$-$01\\
                            & runtime      & 3.8e$+$01 & 1.8e$+$01 & 1.2e$+$01 & 7.5e$+$00 & 1.0e$+$00 & 7.4e$-$01 & 8.3e$+$00 & 1.8e$+$02 & 1.2e$+$03\\\hline
    \multirow{2}{*}{$20$}   & $\hinf$-norm & 4.6e$-$01 & 3.0e$-$01 & 3.7e$-$01 & 3.4e$-$01 & 2.9e$-$01 & 3.2e$-$01 & 4.0e$-$01 & 3.9e$-$01 & 3.8e$-$01\\
                            & runtime      & 5.9e$+$01 & 3.3e$+$01 & 1.4e$+$01 & 7.0e$-$01 & 2.4e$+$00 & 1.6e$+$00 & 2.6e$+$01 & 1.2e$+$02 & 7.2e$+$02\\\hline
    \multirow{2}{*}{$50$}   & $\hinf$-norm & 3.2e$-$01 & 3.8e$-$01 & 3.7e$-$01 & 3.3e$-$01 & 3.1e$-$01 & 3.1e$-$01 & 3.9e$-$01 & 3.8e$-$01 & 3.8e$-$01\\
                            & runtime      & 8.5e$+$01 & 9.3e$+$02 & 3.8e$+$01 & 1.7e$+$00 & 2.7e$+$00 & 4.3e$+$00 & 2.0e$+$01 & 2.0e$+$02 & 1.0e$+$03\\\hline
    \multirow{2}{*}{$100$}  & $\hinf$-norm & 4.5e$-$01 & 3.9e$-$01 & 3.5e$-$01 & 3.2e$-$01 & 3.1e$-$01 & 3.0e$-$01 & 3.9e$-$01 & 3.8e$-$01 & 3.8e$-$01\\
                            & runtime      & 1.2e$+$03 & 3.7e$+$03 & 1.7e$+$02 & 1.4e$+$01 & 1.4e$+$01 & 1.4e$+$01 & 1.8e$+$01 & 4.4e$+$02 & 9.3e$+$02\\\hline
    \multirow{2}{*}{$500$}  & $\hinf$-norm & 4.3e$-$01 & 3.9e$-$01 & 3.5e$-$01 & 3.2e$-$01 & 3.1e$-$01 & 3.0e$-$01 & 3.9e$-$01 & 3.8e$-$01 & 3.8e$-$01\\
                            & runtime      & 5.2e$+$03 & 6.4e$+$03 & 7.9e$+$03 & 9.1e$+$02 & 7.8e$+$02 & 9.7e$+$02 & 1.6e$+$01 & 4.6e$+$02 & 9.7e$+$02\\\hline
    \multirow{2}{*}{$1000$} & $\hinf$-norm & 4.7e$-$01 & 4.0e$-$01 & 4.0e$-$01 & 3.2e$-$01 & 3.1e$-$01 & 3.1e$-$01 & 3.9e$-$01 & 3.8e$-$01 & 3.8e$-$01\\
                            & runtime      & 7.3e$+$03 & 9.3e$+$03 & 8.7e$+$03 & 1.1e$+$04 & 7.8e$+$03 & 8.4e$+$03 & 2.1e$+$01 & 3.7e$+$02 & 8.5e$+$02\\\hline
    \multirow{2}{*}{$2000$} & $\hinf$-norm & --- $ $   & --- $ $   & --- $ $   & 3.3e$-$01 & 3.1e$-$01 & 3.1e$-$01 & 3.9e$-$01 & 3.8e$-$01 & 3.8e$-$01\\
                            & runtime      & --- $ $   & --- $ $   & --- $ $   & 1.7e$+$05 & 1.1e$+$05 & 8.7e$+$04 & 2.0e$+$01 & 5.2e$+$02 & 9.5e$+$02
  \end{tabular}
\end{table*}

\section{Conclusion}

We have presented SOBSYN, a new algorithm for the computation of fixed-order pH 
controllers for pH plants that aim at a low $\hinf$ norm of the resulting 
closed-loop system. The main features of our algorithm in comparison to other 
fixed-order $\hinf$ methods are the sample-based objective function and the 
passivity-based stability guarantee. Both features facilitate the application of 
our method to pH plants with high state-space dimension. Moreover, the 
sample-based nature of the objective function and its passivity-based stability 
also enables the computation of $\hinf$ controllers in a purely data-driven way. 
Therefore, we can compute passive $\hinf$ controllers for passive plants even 
if no access to the plant system matrices is possible but only transfer function 
evaluations are available.

While our adaptive sampling procedure works well in our experiments (and also 
in the much larger set of previously conducted model order reduction 
experiments), it is still possible that sharp peaks  
in the spectral norm of the closed-loop frequency response are missed. Therefore, 
it 
is recommended to validate the final controller performance. For small or 
medium systems, this can be done quickly using well-established $\hinf$ norm 
computation. For large-scale systems, an $\hinf$ certificate developed 
in~\cite{Schwerdtner2020a} may be used. We currently investigate the 
incorporation of such a certificate into the sampling stage of our method.

\section*{Acknowledgements}
We thank Volker Mehrmann for his helpful comments on an earlier version of this manuscript.

\bibliographystyle{elsarticle-harv} 
\bibliography{references}

\begin{thebibliography}{49}
\expandafter\ifx\csname natexlab\endcsname\relax\def\natexlab#1{#1}\fi
\providecommand{\url}[1]{\texttt{#1}}
\providecommand{\href}[2]{#2}
\providecommand{\path}[1]{#1}
\providecommand{\DOIprefix}{doi:}
\providecommand{\ArXivprefix}{arXiv:}
\providecommand{\URLprefix}{URL: }
\providecommand{\Pubmedprefix}{pmid:}
\providecommand{\doi}[1]{\href{http://dx.doi.org/#1}{\path{#1}}}
\providecommand{\Pubmed}[1]{\href{pmid:#1}{\path{#1}}}
\providecommand{\bibinfo}[2]{#2}
\ifx\xfnm\relax \def\xfnm[#1]{\unskip,\space#1}\fi
\bibitem[{Aliyev et~al.(2017)Aliyev, Benner, Mengi, Schwerdtner and
  Voigt}]{Aliyev2017}
\bibinfo{author}{Aliyev, N.}, \bibinfo{author}{Benner, P.},
  \bibinfo{author}{Mengi, E.}, \bibinfo{author}{Schwerdtner, P.},
  \bibinfo{author}{Voigt, M.}, \bibinfo{year}{2017}.
\newblock \bibinfo{title}{Large-scale computation of
  {$\mathcal{L}_{\infty}$}-norms by a greedy subspace method}.
\newblock \bibinfo{journal}{SIAM J. Matrix Anal. Appl.} \bibinfo{volume}{38},
  \bibinfo{pages}{1496--1516}.
\bibitem[{Anderson and Liu(1989)}]{AndL89}
\bibinfo{author}{Anderson, B.D.O.}, \bibinfo{author}{Liu, Y.},
  \bibinfo{year}{1989}.
\newblock \bibinfo{title}{Controller reduction: concepts and approaches}.
\newblock \bibinfo{journal}{IEEE Trans. Automat. Control} \bibinfo{volume}{34},
  \bibinfo{pages}{802--812}.
\bibitem[{Apkarian and Noll(2006)}]{Apkarian2006Hinfstruct}
\bibinfo{author}{Apkarian, P.}, \bibinfo{author}{Noll, D.},
  \bibinfo{year}{2006}.
\newblock \bibinfo{title}{Nonsmooth ${H}_\infty$ synthesis}.
\newblock \bibinfo{journal}{IEEE Trans. Automat. Control} \bibinfo{volume}{51},
  \bibinfo{pages}{71--86}.
\bibitem[{Apkarian and Noll(2018)}]{Apkarian2018}
\bibinfo{author}{Apkarian, P.}, \bibinfo{author}{Noll, D.},
  \bibinfo{year}{2018}.
\newblock \bibinfo{title}{Structured ${H}_\infty$-control of
  infinite-dimensional systems}.
\newblock \bibinfo{journal}{Internat. J. Robust Nonlinear Control}
  \bibinfo{volume}{28}, \bibinfo{pages}{3212--3238}.
\bibitem[{Beattie et~al.(2022)Beattie, Mehrmann and Xu}]{BeaMX22}
\bibinfo{author}{Beattie, C.}, \bibinfo{author}{Mehrmann, V.},
  \bibinfo{author}{Xu, H.}, \bibinfo{year}{2022}.
\newblock \bibinfo{title}{Port-{H}amiltonian Realizations of linear time
  invariant systems}.
\newblock \bibinfo{type}{arXiv preprint} \bibinfo{number}{arXiv:2201.05355}.
\newblock \bibinfo{note}{Available at \url{https://arxiv.org/abs/2201.05355}}.
\bibitem[{Beattie et~al.(2018)Beattie, Mehrmann, Xu and Zwart}]{BeaMXZ18}
\bibinfo{author}{Beattie, C.}, \bibinfo{author}{Mehrmann, V.},
  \bibinfo{author}{Xu, H.}, \bibinfo{author}{Zwart, H.}, \bibinfo{year}{2018}.
\newblock \bibinfo{title}{Linear port-{H}amiltonian descriptor systems}.
\newblock \bibinfo{journal}{Math. Control Signals Systems}
  \bibinfo{volume}{30}, \bibinfo{pages}{17}.
\bibitem[{Benner et~al.(2011)Benner, Byers, Losse, Mehrmann and Xu}]{BenBLMX11}
\bibinfo{author}{Benner, P.}, \bibinfo{author}{Byers, R.},
  \bibinfo{author}{Losse, P.}, \bibinfo{author}{Mehrmann, V.},
  \bibinfo{author}{Xu, H.}, \bibinfo{year}{2011}.
\newblock \bibinfo{title}{Robust formulas for optimal {$H_\infty$}
  controllers}.
\newblock \bibinfo{journal}{Automatica J. IFAC} \bibinfo{volume}{47},
  \bibinfo{pages}{2639--2646}.
\bibitem[{Benner et~al.(2002)Benner, Byers, Mehrmann and Xu}]{BenBMX02}
\bibinfo{author}{Benner, P.}, \bibinfo{author}{Byers, R.},
  \bibinfo{author}{Mehrmann, V.}, \bibinfo{author}{Xu, H.},
  \bibinfo{year}{2002}.
\newblock \bibinfo{title}{Numerical computation of deflating subspaces of
  skew-{H}amiltonian/{H}amiltonian pencils}.
\newblock \bibinfo{journal}{SIAM J. Matrix Anal. Appl.} \bibinfo{volume}{24},
  \bibinfo{pages}{165--190}.
\bibitem[{Benner et~al.(2022)Benner, Heiland and Werner}]{BenHW22b}
\bibinfo{author}{Benner, P.}, \bibinfo{author}{Heiland, J.},
  \bibinfo{author}{Werner, S.W.R.}, \bibinfo{year}{2022}.
\newblock \bibinfo{title}{Robust output-feedback stabilization for
  incompressible flows using low-dimensional $\mathcal{H}_\infty$-controllers}.
\newblock \bibinfo{journal}{Comput. Optim. Appl.} \bibinfo{volume}{82},
  \bibinfo{pages}{225--249}.
\bibitem[{Benner et~al.(2018)Benner, Mitchell and Overton}]{BenMO18}
\bibinfo{author}{Benner, P.}, \bibinfo{author}{Mitchell, T.},
  \bibinfo{author}{Overton, M.L.}, \bibinfo{year}{2018}.
\newblock \bibinfo{title}{Low-order control design using a reduced-order model
  with a stability constraint on the full-order model}, in:
  \bibinfo{booktitle}{Proc. 2018 IEEE Conference on Decision and Control
  (CDC)}, \bibinfo{address}{Miami, FL, USA}. pp. \bibinfo{pages}{3000--3005}.
\bibitem[{Breiten and Karsai(2022)}]{Breiten2022}
\bibinfo{author}{Breiten, T.}, \bibinfo{author}{Karsai, A.},
  \bibinfo{year}{2022}.
\newblock \bibinfo{title}{Structure Preserving ${H}_\infty$ Control for
  Port-{H}amiltonian Systems}.
\newblock \bibinfo{type}{arXiv preprint} \bibinfo{number}{arXiv:2206.08706}.
\newblock \bibinfo{note}{Available at \url{https://arxiv.org/abs/2206.08706}}.
\bibitem[{Burke et~al.(2006)Burke, Henrion, Lewis and Overton}]{Burke2006Hifoo}
\bibinfo{author}{Burke, J.V.}, \bibinfo{author}{Henrion, D.},
  \bibinfo{author}{Lewis, A.S.}, \bibinfo{author}{Overton, M.L.},
  \bibinfo{year}{2006}.
\newblock \bibinfo{title}{{HIFOO} --- a {MATLAB} package for fixed-order
  controller design and ${H}_\infty$ optimization}.
\newblock \bibinfo{journal}{IFAC Proc. Vol.} \bibinfo{volume}{39},
  \bibinfo{pages}{339--344}.
\bibitem[{Cherifi et~al.(2022)Cherifi, Gernandt and Hinsen}]{Cherifi2022}
\bibinfo{author}{Cherifi, K.}, \bibinfo{author}{Gernandt, H.},
  \bibinfo{author}{Hinsen, D.}, \bibinfo{year}{2022}.
\newblock \bibinfo{title}{The difference between port-{H}amiltonian, passive
  and positive real descriptor systems}.
\newblock \bibinfo{note}{Available at \url{https://arxiv.org/abs/2204.04990}}.
\bibitem[{Coelho et~al.(2004)Coelho, Phillips and Silveira}]{Coelho2004}
\bibinfo{author}{Coelho, C.P.}, \bibinfo{author}{Phillips, J.},
  \bibinfo{author}{Silveira, L.M.}, \bibinfo{year}{2004}.
\newblock \bibinfo{title}{A convex programming approach for generating
  guaranteed passive approximations to tabulated frequency-data}.
\newblock \bibinfo{journal}{IEEE Trans. Computer-Aided Des. Integr. Circuits
  Systems} \bibinfo{volume}{23}, \bibinfo{pages}{293--301}.
\bibitem[{Francis(1987)}]{Fra87}
\bibinfo{author}{Francis, B.A.}, \bibinfo{year}{1987}.
\newblock \bibinfo{title}{A Course in $H_\infty$ Control Theory}.
  volume~\bibinfo{volume}{88} of \textit{\bibinfo{series}{Lect. Notes Control
  Inf. Sci.}}
\newblock \bibinfo{publisher}{Spinger}, \bibinfo{address}{Berlin, Heidelberg}.
\bibitem[{Gabarrou et~al.(2010)Gabarrou, Alazard and Noll}]{GabAN10}
\bibinfo{author}{Gabarrou, M.}, \bibinfo{author}{Alazard, D.},
  \bibinfo{author}{Noll, D.}, \bibinfo{year}{2010}.
\newblock \bibinfo{title}{Structured flight control law design using non-smooth
  optimization}.
\newblock \bibinfo{journal}{IFAC Proc. Vol.} \bibinfo{volume}{43},
  \bibinfo{pages}{536--541}.
\bibitem[{Geromel et~al.(1998)Geromel, de~Souza and Skelton}]{GerSS98}
\bibinfo{author}{Geromel, J.C.}, \bibinfo{author}{de~Souza, C.C.},
  \bibinfo{author}{Skelton, R.E.}, \bibinfo{year}{1998}.
\newblock \bibinfo{title}{Static output feedback controllers: Stability and
  convexity}.
\newblock \bibinfo{journal}{IEEE Trans. Automat. Control} \bibinfo{volume}{43},
  \bibinfo{pages}{120--125}.
\bibitem[{Gillis and Sharma(2018)}]{Gillis2018}
\bibinfo{author}{Gillis, N.}, \bibinfo{author}{Sharma, P.},
  \bibinfo{year}{2018}.
\newblock \bibinfo{title}{Finding the nearest positive-real system}.
\newblock \bibinfo{journal}{SIAM Journal on Numerical Analysis}
  \bibinfo{volume}{56}, \bibinfo{pages}{1022--1047}.
\newblock \DOIprefix\doi{10.1137/17M1137176}.
\bibitem[{Grivet-Talocia(2004)}]{Grivet2004}
\bibinfo{author}{Grivet-Talocia, S.}, \bibinfo{year}{2004}.
\newblock \bibinfo{title}{Passivity enforcement via perturbation of
  {H}amiltonian matrices}.
\newblock \bibinfo{journal}{IEEE Trans. Circuits Syst. I. Regul. Pap.}
  \bibinfo{volume}{51}, \bibinfo{pages}{1755--1769}.
\bibitem[{Grivet-Talocia and Gustavsen(2015)}]{Grivet2015}
\bibinfo{author}{Grivet-Talocia, S.}, \bibinfo{author}{Gustavsen, B.},
  \bibinfo{year}{2015}.
\newblock \bibinfo{title}{Passive Macromodeling}.
\newblock Wiley Ser. Microwave Optical Engrg., \bibinfo{publisher}{Wiley},
  \bibinfo{address}{Nashville, TN}.
\bibitem[{Gugercin et~al.(2012)Gugercin, Polyuga, Beattie and van~der
  Schaft}]{Gugercin2012}
\bibinfo{author}{Gugercin, S.}, \bibinfo{author}{Polyuga, R.V.},
  \bibinfo{author}{Beattie, C.}, \bibinfo{author}{van~der Schaft, A.},
  \bibinfo{year}{2012}.
\newblock \bibinfo{title}{Structure-preserving tangential interpolation for
  model reduction of port-{H}amiltonian systems}.
\newblock \bibinfo{journal}{Automatica J. IFAC} \bibinfo{volume}{48},
  \bibinfo{pages}{1963--1974}.
\bibitem[{Guglielmi et~al.(2013)Guglielmi, G\"urb\"uzbalaban and
  Overton}]{GugGO13}
\bibinfo{author}{Guglielmi, N.}, \bibinfo{author}{G\"urb\"uzbalaban, M.},
  \bibinfo{author}{Overton, M.L.}, \bibinfo{year}{2013}.
\newblock \bibinfo{title}{Fast approximation of the ${H_\infty}$ norm via
  optimization over spectral value sets}.
\newblock \bibinfo{journal}{SIAM J. Matrix Anal. Appl.} \bibinfo{volume}{34},
  \bibinfo{pages}{709--737}.
\bibitem[{Gustavsen and Semlyen(2001)}]{Gus2001}
\bibinfo{author}{Gustavsen, B.}, \bibinfo{author}{Semlyen, A.},
  \bibinfo{year}{2001}.
\newblock \bibinfo{title}{Enforcing passivity for admittance matrices
  approximated by rational functions}.
\newblock \bibinfo{journal}{IEEE Trans. Power Syst.} \bibinfo{volume}{16},
  \bibinfo{pages}{97--104}.
\bibitem[{Hauschild et~al.(2020)Hauschild, Marheineke, Mehrmann, Mohring,
  Badlyan, Rein and Schmidt}]{HauMMMMRS19}
\bibinfo{author}{Hauschild, S.A.}, \bibinfo{author}{Marheineke, N.},
  \bibinfo{author}{Mehrmann, V.}, \bibinfo{author}{Mohring, J.},
  \bibinfo{author}{Badlyan, A.M.}, \bibinfo{author}{Rein, M.},
  \bibinfo{author}{Schmidt, M.}, \bibinfo{year}{2020}.
\newblock \bibinfo{title}{Port-{H}amiltonian modeling of district heating
  networks}, in: \bibinfo{editor}{Reis, T.}, \bibinfo{editor}{Grundel, S.},
  \bibinfo{editor}{Sch\"ops, S.} (Eds.), \bibinfo{booktitle}{Progress in
  Differential-Algebraic Equations~II}. \bibinfo{publisher}{Spinger},
  \bibinfo{address}{Cham}. Differ.-Algebr. Equ. Forum, pp.
  \bibinfo{pages}{333--355}.
\bibitem[{Khalil(2002)}]{Khalil2002}
\bibinfo{author}{Khalil, H.K.}, \bibinfo{year}{2002}.
\newblock \bibinfo{title}{{Nonlinear Systems}}.
\newblock \bibinfo{publisher}{Prentice-Hall}, \bibinfo{address}{Upper Saddle
  River, NJ}.
\bibitem[{{Kunkel} and {Mehrmann}(2006)}]{KunM06}
\bibinfo{author}{{Kunkel}, P.}, \bibinfo{author}{{Mehrmann}, V.},
  \bibinfo{year}{2006}.
\newblock \bibinfo{title}{Differential-Algebraic Equations: Analysis and
  Numerical Solution}.
\newblock \bibinfo{publisher}{EMS Publishing House},
  \bibinfo{address}{Z\"urich}.
\bibitem[{Leibfritz(2004)}]{Leibfritz2006}
\bibinfo{author}{Leibfritz, F.}, \bibinfo{year}{2004}.
\newblock \bibinfo{title}{\textit{COMPl$_e$ib}: ${CO}$nstrained
  ${M}$atrix–optimization ${P}$roblem $li$brary --- a collection of test
  examples for nonlinear semidefinite programs, control system design and
  related problems}.
\newblock \bibinfo{note}{Available at
  \url{http://www.friedemann-leibfritz.de/COMPlib_Data/COMPlib_Main_Paper.pdf}}.
\bibitem[{Mehl et~al.(2021)Mehl, Mehrmann and Wojtylak}]{MehMW21}
\bibinfo{author}{Mehl, C.}, \bibinfo{author}{Mehrmann, V.},
  \bibinfo{author}{Wojtylak, M.}, \bibinfo{year}{2021}.
\newblock \bibinfo{title}{Distance problems for dissipative {H}amiltonian
  systems and related matrix polynomials}.
\newblock \bibinfo{journal}{Linear Algebra Appl.} \bibinfo{volume}{623},
  \bibinfo{pages}{335--366}.
\bibitem[{Mehrmann et~al.(2018)Mehrmann, Morandin, Olmi and
  Sch\"oll}]{MehMOS18}
\bibinfo{author}{Mehrmann, V.}, \bibinfo{author}{Morandin, R.},
  \bibinfo{author}{Olmi, S.}, \bibinfo{author}{Sch\"oll, E.},
  \bibinfo{year}{2018}.
\newblock \bibinfo{title}{Qualitative stability and synchronicity analysis of
  power network models in port-{H}amiltonian form}.
\newblock \bibinfo{journal}{Chaos} \bibinfo{volume}{28},
  \bibinfo{pages}{101102}.
\bibitem[{Mehrmann and Unger(2022)}]{MehU22}
\bibinfo{author}{Mehrmann, V.}, \bibinfo{author}{Unger, B.},
  \bibinfo{year}{2022}.
\newblock \bibinfo{title}{Control of port-{H}amiltonian differential-algebraic
  systems and applications}.
\newblock \bibinfo{note}{Available at \url{https://arxiv.org/abs/2201.06590}}.
\bibitem[{Mitchell and Overton(2015)}]{MitO15}
\bibinfo{author}{Mitchell, T.}, \bibinfo{author}{Overton, M.L.},
  \bibinfo{year}{2015}.
\newblock \bibinfo{title}{Fixed low-order controller design and {$H_\infty$}
  optimization for large-scale dynamical systems}.
\newblock \bibinfo{journal}{IFAC-PapersOnLine} \bibinfo{volume}{48},
  \bibinfo{pages}{25--30}.
\bibitem[{Mitchell and Overton(2016)}]{MitO16}
\bibinfo{author}{Mitchell, T.}, \bibinfo{author}{Overton, M.L.},
  \bibinfo{year}{2016}.
\newblock \bibinfo{title}{Hybrid expansion-contraction: a robust scaleable
  method for approximating the {$H_\infty$} norm}.
\newblock \bibinfo{journal}{IMA J. Numer. Anal.} \bibinfo{volume}{36},
  \bibinfo{pages}{985--1014}.
\bibitem[{Mogensen and Riseth(2018)}]{mogensen2018optim}
\bibinfo{author}{Mogensen, P.K.}, \bibinfo{author}{Riseth, A.N.},
  \bibinfo{year}{2018}.
\newblock \bibinfo{title}{Optim: A mathematical optimization package for
  {Julia}}.
\newblock \bibinfo{journal}{J. Open Source Softw.} \bibinfo{volume}{3},
  \bibinfo{pages}{615--618}.
\bibitem[{Moser et~al.(2022)Moser, Schwerdtner, Mehrmann and
  Voigt}]{MosSMV2022b}
\bibinfo{author}{Moser, T.}, \bibinfo{author}{Schwerdtner, P.},
  \bibinfo{author}{Mehrmann, V.}, \bibinfo{author}{Voigt, M.},
  \bibinfo{year}{2022}.
\newblock \bibinfo{title}{Structure-Preserving Model Order Reduction for Index
  Two Port-Hamiltonian Descriptor Systems}.
\newblock \bibinfo{type}{arXiv Preprint} \bibinfo{number}{arXiv:2206.03942}.
\newblock \bibinfo{note}{Available at \url{https://arxiv.org/abs/2206.03942}}.
\bibitem[{Mustafa and Glover(1991)}]{MusG91}
\bibinfo{author}{Mustafa, D.}, \bibinfo{author}{Glover, K.},
  \bibinfo{year}{1991}.
\newblock \bibinfo{title}{Controller reduction by $\mathscr{H}_\infty$-balanced
  trunction}.
\newblock \bibinfo{journal}{IEEE Trans. Automat. Control} \bibinfo{volume}{36},
  \bibinfo{pages}{668--682}.
\bibitem[{O'Donoghue et~al.(2016)O'Donoghue, Chu, Parikh and Boyd}]{OCPB16}
\bibinfo{author}{O'Donoghue, B.}, \bibinfo{author}{Chu, E.},
  \bibinfo{author}{Parikh, N.}, \bibinfo{author}{Boyd, S.},
  \bibinfo{year}{2016}.
\newblock \bibinfo{title}{Conic optimization via operator splitting and
  homogeneous self-dual embedding}.
\newblock \bibinfo{journal}{J. Optim. Theory Appl.} \bibinfo{volume}{169},
  \bibinfo{pages}{1042--1068}.
\bibitem[{Oliveira et~al.(2016)Oliveira, Rodier and Ihlenfeld}]{Oliveira2014}
\bibinfo{author}{Oliveira, G.H.C.}, \bibinfo{author}{Rodier, C.},
  \bibinfo{author}{Ihlenfeld, L.P.R.K.}, \bibinfo{year}{2016}.
\newblock \bibinfo{title}{{LMI}-based method for estimating passive blackbox
  models in power systems transient analysis}.
\newblock \bibinfo{journal}{IEEE Trans. Power Delivery} \bibinfo{volume}{31},
  \bibinfo{pages}{3--10}.
\bibitem[{Ortega et~al.(2008)Ortega, van~der Schaft, Casta{\~n}os and
  Astolfi}]{OrtVCA08}
\bibinfo{author}{Ortega, R.}, \bibinfo{author}{van~der Schaft, A.},
  \bibinfo{author}{Casta{\~n}os, F.}, \bibinfo{author}{Astolfi, A.},
  \bibinfo{year}{2008}.
\newblock \bibinfo{title}{Control by interconnection and standard
  passivity-based control of port-{H}amiltonian systems}.
\newblock \bibinfo{journal}{IEEE Trans. Automat. Control} \bibinfo{volume}{53},
  \bibinfo{pages}{2527--2542}.
\bibitem[{Ramirez et~al.(2016)Ramirez, Gorrec, Maschke and Couenne}]{RamLMC16}
\bibinfo{author}{Ramirez, H.}, \bibinfo{author}{Gorrec, Y.L.},
  \bibinfo{author}{Maschke, B.}, \bibinfo{author}{Couenne, F.},
  \bibinfo{year}{2016}.
\newblock \bibinfo{title}{On the passivity based control of irreversible
  processes: A port-{H}amiltonian approach}.
\newblock \bibinfo{journal}{Automatica J. IFAC} \bibinfo{volume}{64},
  \bibinfo{pages}{105--111}.
\bibitem[{Ravanbod and Noll(2012)}]{RavN12}
\bibinfo{author}{Ravanbod, L.}, \bibinfo{author}{Noll, D.},
  \bibinfo{year}{2012}.
\newblock \bibinfo{title}{Gain-scheduled two-loop autopilot for an aircraft}.
\newblock \bibinfo{journal}{IFAC Proc. Vol.} \bibinfo{volume}{45},
  \bibinfo{pages}{772--777}.
\bibitem[{Robu et~al.(2010)Robu, Budinger, Baudouin, Prieur and
  Arzelier}]{RobBBPA10}
\bibinfo{author}{Robu, B.}, \bibinfo{author}{Budinger, V.},
  \bibinfo{author}{Baudouin, L.}, \bibinfo{author}{Prieur, C.},
  \bibinfo{author}{Arzelier, D.}, \bibinfo{year}{2010}.
\newblock \bibinfo{title}{Simultaneous {$H_\infty$} vibration control of
  fluid/plate system via reduced-order controller}, in:
  \bibinfo{booktitle}{Proc. 49th IEEE Conference on Decision and Control
  (CDC)}, \bibinfo{address}{Atlanta, GA, USA}. pp. \bibinfo{pages}{3146--3151}.
\bibitem[{Schwerdtner et~al.(2020)Schwerdtner, Mengi and
  Voigt}]{Schwerdtner2020a}
\bibinfo{author}{Schwerdtner, P.}, \bibinfo{author}{Mengi, E.},
  \bibinfo{author}{Voigt, M.}, \bibinfo{year}{2020}.
\newblock \bibinfo{title}{Certifying global optimality for the
  {$\mathcal{L}_\infty$}-norm computation of large-scale descriptor systems}.
\newblock \bibinfo{journal}{{IFAC}-{PapersOnLine}} \bibinfo{volume}{53},
  \bibinfo{pages}{4279--4284}.
\bibitem[{Schwerdtner and Voigt(2020)}]{SchV22}
\bibinfo{author}{Schwerdtner, P.}, \bibinfo{author}{Voigt, M.},
  \bibinfo{year}{2020}.
\newblock \bibinfo{title}{{SOBMOR:} {S}tructured Optimization-Based Model Order
  Reduction}.
\newblock \bibinfo{type}{arXiv preprint} \bibinfo{number}{arXiv:2011.07567}.
\newblock \bibinfo{note}{Available at
  \url{https://arxiv.org/pdf/2011.07567.pdf}}.
\bibitem[{Schwerdtner and Voigt(2021)}]{SchV21}
\bibinfo{author}{Schwerdtner, P.}, \bibinfo{author}{Voigt, M.},
  \bibinfo{year}{2021}.
\newblock \bibinfo{title}{Adaptive sampling for structure-preserving model
  order reduction of port-{H}amiltonian systems}.
\newblock \bibinfo{journal}{IFAC-PapersOnLine} \bibinfo{volume}{54},
  \bibinfo{pages}{143--148}.
\bibitem[{Skogestad and Postlethwaite(2005)}]{SkoP05}
\bibinfo{author}{Skogestad, S.}, \bibinfo{author}{Postlethwaite, I.},
  \bibinfo{year}{2005}.
\newblock \bibinfo{title}{Multivariable Feedback Control}.
\newblock \bibinfo{edition}{2nd} ed., \bibinfo{publisher}{Wiley}.
\bibitem[{Wang and Chen(2009)}]{WanC09}
\bibinfo{author}{Wang, F.C.}, \bibinfo{author}{Chen, H.T.},
  \bibinfo{year}{2009}.
\newblock \bibinfo{title}{Design and implementation of fixed-order robust
  controllers for a proton exchange membrane fuel cell system}.
\newblock \bibinfo{journal}{Internat. J. Hydrogen Energy} \bibinfo{volume}{34},
  \bibinfo{pages}{2705--2717}.
\bibitem[{Werner et~al.(2022)Werner, Overton and Peherstorfer}]{WerOP22}
\bibinfo{author}{Werner, S.W.R.}, \bibinfo{author}{Overton, M.L.},
  \bibinfo{author}{Peherstorfer, B.}, \bibinfo{year}{2022}.
\newblock \bibinfo{title}{Multi-fidelity robust controller design with gradient
  sampling}.
\newblock \bibinfo{type}{arXiv preprint} \bibinfo{number}{arXiv:2205.15050}.
\newblock \bibinfo{note}{Available at \url{https://arxiv.org/abs/2205.15050}}.
\bibitem[{Zhang et~al.(2017)Zhang, Borja, Ortega, Liu and Su}]{ZhaBOLS17}
\bibinfo{author}{Zhang, M.}, \bibinfo{author}{Borja, P.},
  \bibinfo{author}{Ortega, R.}, \bibinfo{author}{Liu, Z.}, \bibinfo{author}{Su,
  H.}, \bibinfo{year}{2017}.
\newblock \bibinfo{title}{{PID} passivity-based control of port-{H}amiltonian
  systems}.
\newblock \bibinfo{journal}{IEEE Trans. Automat. Control} \bibinfo{volume}{63},
  \bibinfo{pages}{1032--1044}.
\bibitem[{Zhou et~al.(1996)Zhou, Doyle and Glover}]{ZhoDG96}
\bibinfo{author}{Zhou, K.}, \bibinfo{author}{Doyle, J.C.},
  \bibinfo{author}{Glover, K.}, \bibinfo{year}{1996}.
\newblock \bibinfo{title}{Robust and Optimal Control}.
\newblock \bibinfo{publisher}{Prentice-Hall}, \bibinfo{address}{Englewood
  Cliffs, NJ}.

\end{thebibliography}


\end{document}